\documentclass[journal]{IEEEtran}
\usepackage{amsmath,amssymb,amsfonts}
\usepackage{array}
\usepackage[caption=false,font=normalsize,labelfont=sf,textfont=sf]{subfig}
\usepackage[group-separator={,},group-minimum-digits=3]{siunitx}
\usepackage{textcomp}
\usepackage{stfloats}
\usepackage{url}
\usepackage{verbatim}
\usepackage{epstopdf}
\usepackage{graphicx}
\usepackage{balance}
\usepackage{mathrsfs}
\usepackage{cite}
\usepackage{xcolor}
\usepackage{booktabs}
\usepackage{threeparttable}
\usepackage{amsthm}          
\usepackage{algorithmicx}
\usepackage[noend]{algpseudocode}

\usepackage[linewidth=1pt]{mdframed}
\usepackage[colorlinks=true,
            linkcolor=blue,
            anchorcolor=blue,
            citecolor=blue]{hyperref}
\begin{document}

\title{\fontsize{16pt}{\baselineskip}\selectfont 
 {Unsupervised Semi-Parametric Plug-in Likelihood-Ratio Detection for Covert Communications in the Presence of Disco Reconfigurable Intelligent Surfaces}
} 
\author{ 
{
    Luyao~Sun,
    Huan~Huang,~\textit{Member,~IEEE},
    Yongxing~Song,
    Zhongxing~Tian,
    Hongliang~Zhang,~\textit{Member,~IEEE},
    Weidong~Mei,~\textit{Member,~IEEE}, 
    Dongdong~Zou,
    and Yi~Cai
}
\thanks{  

L.~Sun, H.~Huang, Y.~Song, Z.~Tian, D.~Zou,  Y.~Cai
are with the School of Electronic and Information Engineering, Soochow University, Suzhou, Jiangsu 215006, China (e-mail: lysun02@163.com; hhuang1799@gmail.com; yxsong03@163.com; zxtian@ieee.org; ddzou@suda.edu.cn; yicai@ieee.org). 

H.~Zhang is with the School of Electronics, Peking University, Beijing 100871, China 
(e-mail: hongliang.zhang92@gmail.com).

W. Mei is with the National Key Laboratory of Wireless Communications, University of Electronic Science and Technology of China, Chengdu 611731, China (e-mail: wmei@uestc.edu.cn).

}
}
\maketitle

\begin{abstract}
Covert communications, also referred to as low probability of detection (LPD) communications, provide a higher level of privacy protection than cryptography and physical-layer security (PLS) by hiding transmissions in the ambient environment.
In this work, we investigate covert communications in the presence of a disco reconfigurable intelligent surface (DRIS) deployed by the warden Willie, which reduces Willie's detection error probability (DEP), i.e., the sum of the false alarm rate (FAR) and the miss detection rate (MDR), and degrades the communication performance between Alice and Bob, without relying on either channel state information (CSI) or additional jamming power. 
However, the introduction of the DRIS makes it analytically intractable for Willie to construct the Neyman--Pearson (NP) detector, which is the optimal detector for monitoring potential covert transmissions between Alice and Bob.
To this end, we develop an unsupervised semi-parametric plug-in likelihood-ratio detector for Willie. 
The proposed detector retains the parametric Gamma reference model under the silent hypothesis without requiring prior knowledge of noise, and learns from unlabeled data a one-dimensional monotone normalizing flow model for the analytically intractable distribution under the transmission hypothesis. 
In particular, it exploits the structural prior inherent in covert communications that Willie's observations reduce to noise only when Alice and Bob are silent. The monitoring performance at Willie is evaluated in terms of DEP, while the communication impact on Alice and Bob is quantified by the signal-to-jamming-plus-noise ratio (SJNR). Simulation results verify the analysis and show that the proposed unsupervised plug-in likelihood-ratio detector achieves monitoring performance close to that of its supervised counterpart.
\end{abstract}
 
\begin{IEEEkeywords}
Covert communications, reconfigurable intelligent surface, physical layer security, normalizing flow, likelihood-ratio detection, channel aging.
\end{IEEEkeywords}

\section{Introduction}\label{Intro}
Due to the broadcast and superposition properties of wireless channels, wireless communication systems are vulnerable to malicious attacks~\cite{PLSsur1}. This vulnerability is particularly critical in the context of the Internet of Things (IoT), where transmitted signals may carry sensitive personal information, such as healthcare records and location information. Communication confidentiality is also imperative in security-sensitive applications, particularly in government and military scenarios. As a result, considerable research attention has been devoted to developing techniques that improve transmission security and protect user privacy~\cite{CCRef1}. Covert communications have emerged as a promising strategy~\cite{CCRef2,CCRef3,CCRef4}, aiming to conceal the very existence of transmissions by embedding them in ambient noise or legitimate traffic, thereby significantly reducing their probability of detection. Unlike conventional cryptographic techniques and physical-layer security (PLS) mechanisms, which primarily protect the confidentiality of transmitted content, covert communications offer a higher level of privacy protection by concealing whether communication occurs at all. 
 
\subsection{Related Works}
In~\cite{CCRef5}, a fundamental limit of covert communications was established: over $n$ additive white Gaussian noise (AWGN) channel uses, at most $o(\sqrt{n})$ bits can be transmitted while maintaining an arbitrarily low probability of detection (LPD) when the noise power on the channel between the transmitter Alice and the warden Willie is unknown, where $o(\sqrt{n})$ denotes a non-asymptotically tight upper bound on $\sqrt{n}$. Furthermore, if a lower bound on the noise power is available, the achievable covert throughput increases to $\mathcal{O}(\sqrt{n})$ bits, where $\mathcal{O}(\sqrt{n})$ is an asymptotically tight upper bound on $\sqrt{n}$.
Following the work in~\cite{CCRef5}, existing studies introduced techniques such as non-orthogonal multiple access (NOMA) and turbo encoding to improve covert communication performance~\cite{CCwork1,CCwork2}. 
Existing studies~\cite{CCwork5,CCwork6,CCwork7} examined the use of jammers or artificial noise to increase Willie's observation uncertainty, thereby making it harder for Willie to reliably detect Alice's activity.
Relay-assisted schemes~\cite{CCwork8,CCRelay11,CCRelay2} can also reduce the separability between Willie's two hypotheses, thus increasing his detection error probabilities.
In addition, the work in~\cite{CCwork11} proposed exploiting small-scale-fading-induced power variability to reduce Willie's detection reliability and attain LPD.

Recently, reconfigurable intelligent surfaces (RISs) have emerged as a promising approach for improving the performance and efficiency of wireless communication systems~\cite{CRIS1,CRIS2,CRIS3,CRIS4,CRIS5,CRIS6}. 
An RIS typically comprises a large number of reflecting elements that can dynamically adjust the amplitude and phase of incident signals through low-cost components such as PIN or varactor diodes~\cite{CRIS_coding}. 
The integration of RISs into wireless systems provides considerable performance gains without incurring significant increases in power consumption or deployment costs~\cite{CRIS_energy,CRIS_dl}. 
Moreover, RIS-assisted covert communications have attracted growing research interest~\cite{C_C_RIS1,C_C_RIS2,C_C_RIS3,C_C_RIS4,C_C_RIS5}. 
Previous studies have explored the use of single or multiple RISs to enhance covertness, including those with NOMA~\cite{C_C_RIS1}, artificial noise~\cite{C_C_RIS2}, finite blocklength coding with adaptive priors~\cite{C_C_RIS3}, and unmanned aerial vehicle (UAV)-based networks~\cite{C_C_RIS4,C_C_RIS5}.

Existing research on covert communications, with or without RISs, generally assumes channel reciprocity in time-division duplex (TDD) systems and treats it as exact or approximately valid. Although it was widely adopted, this assumption no longer holds when using time-varying disco RIS (DRIS)~\cite{MyWCMag} (i.e., an RIS with random and time-varying configurations) or RIS architectures with non-reciprocal inter-element connections~\cite{CRIS_Channel_reciprocity}. The use of DRIS in covert communications was first proposed in~\cite{huang2025simultaneously}, where a random and time-varying DRIS that acts like a ``disco ball'' was implemented to generate active channel aging (ACA). Consequently, it breaks TDD channel reciprocity, even within the channel coherence time. A DRIS can be implemented without requiring either channel state information (CSI) or additional jamming power. At the same time, it reduces Willie's probability of detection and degrades the communication performance between Alice and Bob~\cite{huang2025simultaneously}. The introduction of DRIS faces the following two fundamental challenges of conventional covert communications: (i) TDD channel reciprocity is intentionally aged even within a channel coherence interval, rendering existing threshold designs based on static statistical models inapplicable; and (ii) the distribution of the test statistic under Alice--Bob transmission (i.e., under hypothesis $\mathcal{H}_1$) no longer admits a closed-form expression.

Normalizing flows provide a principled route to likelihood-based density modeling by transforming a simple base distribution into a target distribution through an invertible mapping with a tractable Jacobian determinant~\cite{NF_RN,NF_RealNVP}. As a result, they enable exact likelihood evaluation for complex distributions. In covert communications, the physical metric at Willie used for monitoring is typically receive power~\cite{CCRef3,C_C_RIS1,C_C_RIS4}, and thus the object needed to be modeled by Willie is not a high-dimensional received-signal vector but the scalar test statistic under $\mathcal{H}_1$. Hence, what is required here is not a generic high-dimensional flow architecture, but a flexible one-dimensional monotone normalizing flow that can represent the analytically intractable density of the test statistic while preserving exact change-of-variables likelihood computation~\cite{NF_NSF}. If the density under $\mathcal{H}_1$ were available, Willie could in principle construct the corresponding likelihood-ratio detector even in the presence of DRIS-induced ACA.

Unfortunately, there are still two obstacles due to the introduction of DRISs. First, the distribution of the test statistic under $\mathcal{H}_1$ does not admit a closed-form expression. Second, the warden Willie has no access to  labeled training samples in covert communications, i.e., a scenario that is typically adversarial. Therefore, a likelihood detector based on supervised normalizing flows would rely on side information that Willie does not possess, since each observation would need to be tagged according to whether it was generated during Alice--Bob transmission ($\mathcal{H}_1$) or silence ($\mathcal{H}_0$). These considerations motivate an unsupervised semi-parametric plug-in likelihood-ratio detector developed in this work. 
Specifically, Willie preserves the parametric Gamma reference model under $\mathcal{H}_0$, jointly estimates its effective scale and a flexible normalizing flow-based model for the intractable density under $\mathcal{H}_1$ from unlabeled sensing data, and then forms the resulting plug-in likelihood ratio for detection.

\subsection{Contributions and Organization}
In this work, we develop the unsupervised semi-parametric plug-in likelihood-ratio detector for the warden Willie to monitor covert communications between Alice and Bob in the presence of a DRIS. The main contributions of this paper are summarized as follows.
\begin{itemize}
\item
We formulate a covert communication model in the presence of a DRIS, where the DRIS with random and time-varying reflection coefficients acts like a `disco ball'' and operates autonomously without any connection or coordination with the warden Willie. The introduction of the DRIS not only jams the covert communications between Alice and Bob but also reduces Willie's detection error probability (DEP), i.e., the sum of the false alarm rate (FAR) and the miss detection rate (MDR), without relying on either Alice--Bob channel knowledge or additional jamming power. To quantify the impact of the DRIS on covert communications, we use the FAR and MDR as performance metrics for Willie, and the signal-to-jamming-plus-noise ratio (SJNR) as the communication performance metric at Bob.

\item
For covert communications in the presence of a DRIS, the probability density function (PDF) of the test statistic at Willie is analytically intractable under ${\mathcal{H}}_1$. As a result, the exact Neyman--Pearson (NP) detector cannot be constructed in closed form. Meanwhile, due to the adversarial relationship between Willie and Alice/Bob, it is unrealistic to assume that Willie has access to a labeled training dataset that distinguishes samples collected under ${\mathcal{H}}_1$ and ${\mathcal{H}}_0$, or that he has access to prior knowledge of noise. To overcome these difficulties, we develop an unsupervised semi-parametric plug-in likelihood-ratio detector. The proposed detector retains the parametric Gamma reference model under ${\mathcal H}_0$ without requiring oracle knowledge of the noise power, while learning the intractable density under ${\mathcal H}_1$ from unlabeled data through a flexible one-dimensional monotone normalizing flow model. Based on this construction, Willie can form a plug-in log-likelihood ratio (LLR) for detection without requiring oracle side information.

\item
The DRIS deployed by Willie can also significantly degrade the SJNR at Bob, thereby giving Willie a strategic advantage whenever covert communication occurs, even if Willie misdetects. Based on the derived statistical characteristic of the cascaded DRIS channel, we conduct an asymptotic analysis of the SJNR to quantify the DRIS impact on covert communications between Alice and Bob. 
Simulation results validate the asymptotic analysis and show that the proposed unsupervised semi-parametric plug-in likelihood-ratio detector achieves monitoring performance close to that of its supervised counterpart. The results also reveal several interesting properties of DRIS on covert communications. For example, increasing Alice's transmit power does not significantly improve Bob's SJNR in the presence of the DRIS; instead, it strengthens the DRIS-induced ACA and increases Alice's risk of detection by Willie. In addition, even a 1-bit phase quantization DRIS is sufficient to improve Willie's monitoring accuracy and degrade Alice/Bob's communication performance.
\end{itemize}

The rest of this paper is organized as follows. In Section~\ref{SystemModel}, we model covert communications in the presence of a DRIS and establish wireless channels. We then derive the decision rule for Willie, introduce the DEP as the monitoring performance metric, and use the SJNR to quantify the communication performance metric. Some useful preliminary results on hypothesis testing and random variables are also presented. In Section~\ref{AIDec}, we describe the challenges introduced by the DRIS for Willie's detection, where the PDF of the test statistic under ${\mathcal{H}}_1$ is analytically intractable. We then present the proposed unsupervised semi-parametric plug-in likelihood-ratio detector, which retains the parametric Gamma reference model under ${\mathcal{H}}_0$ without requiring the noise variance, and learns from unlabeled data a one-dimensional monotone normalizing flow model for the analytically intractable density of the test statistic under ${\mathcal{H}}_1$. In addition, we derive the statistical characteristics of DRIS-induced ACA and conduct an asymptotic analysis of the SJNR at Bob to quantify the communication performance between Alice and Bob. In Section~\ref{Simu}, simulation results are presented to validate the theoretical analysis and verify the effectiveness of the proposed unsupervised detector. Finally, conclusions are drawn in Section~\ref{conclusions}.

\emph{Notation:} 
We use lowercase boldface letters for vectors (e.g., ${\boldsymbol g}$) 
and italic letters for scalars (e.g., $N_{\rm D}$). 
The operators $(\cdot)^{T}$ and $(\cdot)^{H}$ denote the transpose 
and the Hermitian transpose, respectively. 
The symbol $|\cdot|$ denotes absolute value, 
and $\mathbb{E}\!\left[\cdot\right]$ denotes statistical expectation.

\section{System Description}\label{SystemModel}
In Section~\ref{CCModel}, we first present covert communications in the presence of a DRIS 
and build the wireless channel models.
In Section~\ref{PerformanceMer}, the decision rule adopted by the warden Willie is first defined.
Then, the DEP, i.e., the sum of FAR and MDR are given as monitoring performance metrics for Willie, while the SJNR is defined as a communication performance metric.
In Section~\ref{RRs}, some useful results on signal detections and random variables are presented.

\subsection{Covert Communications in the Presence of a DRIS}\label{CCModel}
Fig.~\ref{fig1} illustrates a covert communication system in the presence of a DRIS, where Willie attempts to more effectively detect and jam potential covert communications between Alice and Bob by implementing a DRIS. In covert communications, Alice aims to covertly transmit messages to Bob without being detected by the warden Willie while maintaining communication quality~\cite{CCRef4,CCRef5}. In contrast, Willie aims to detect the covert communications between Alice and Bob as accurately as possible. In this work, we consider a covert communication scenario in which Willie has no knowledge of the Alice--Bob channels, including Bob's location. Therefore, Willie attempts to detect potential covert communications without any coordination with Alice or Bob. Moreover, we assume that Willie does not have access to a labeled dataset for training a supervised detection model.

\begin{figure*}[!t]
\centering
\includegraphics[scale=0.66]{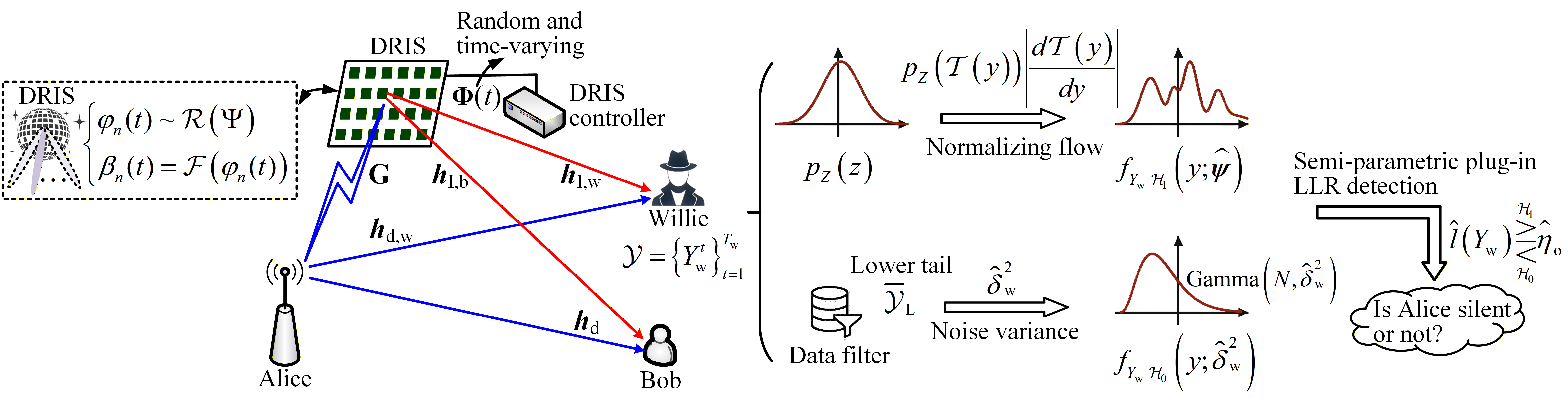}
\caption{Covert communications in the presence of a disco reconfigurable intelligent surface (DRIS), where the warden Willie employs the masked autoregressive flow to assist in his detection and the DRIS with time-varying and random reflection coefficients is generated by a DRIS controller without any connection or coordination with Willie.}
\label{fig1}
\end{figure*}

The warden Willie employs an $N_{\rm D}$-element DRIS with random and time-varying coefficient configurations to improve the detection accuracy of covert communications. Notably, even when Willie fails to detect the covert transmission, the DRIS can still impose significant ACA on the communication between Alice and Bob. The resulting DRIS-induced ACA is analyzed in subsequent sections. The DRIS coefficients are controlled by programmable PINs~\cite{CRIS_coding}, whose ON/OFF behavior only permits discrete coefficient adjustment. More specifically, we assume that the DRIS employs $b$-bit quantized phase shifts and amplitudes. The phase-shift and amplitude sets are denoted by $\Omega=\left\{\phi_1,\phi_2,\cdots,\phi_{2^b}\right\}$ and $\Lambda=\left\{\alpha_1,\alpha_2,\cdots,\alpha_{2^b}\right\}$,
respectively. Accordingly, the random and time-varying DRIS reflection vector is expressed as
${\boldsymbol \varphi}(t)=\left[\beta_1(t)e^{j\varphi_1(t)},\beta_2(t)e^{j\varphi_2(t)},\cdots,\beta_{N_{\rm D}}(t)e^{j\varphi_{N_{\rm D}}(t)}\right]$, where the phase shift of the $r$-th DRIS element is randomly generated according to a distribution $\mathcal{R}$ over $\Omega$, i.e., $\varphi_r(t)\sim\mathcal{R}(\Omega)$. In practice, the amplitude $\beta_r(t)$, $r=1,2,\cdots,N_{\rm D}$, typically depends on the corresponding phase shift and can be modeled as $\beta_r(t)=\mathcal{F}(\varphi_r(t))$, where $\mathcal{F}$ is a deterministic mapping specified by the hardware characteristics.

In covert communications, Willie attempts to determine from his observations which of the following two hypotheses holds: Alice is transmitting (${\mathcal H}_1$), or Alice is silent (${\mathcal H}_0$). Referring to the covert communication model in~\cite{C_C_RIS1}, the $m$-th received signal at Willie can be expressed as
\begin{equation}
    y_{\rm w}(m)=
    \left\{
    \begin{array}{ll}
        \underbrace{h^{\rm w}_{\rm d}s(m)}_{\rm Direct\;link}
        +\underbrace{h^{\rm w}_{\rm D}(m)s(m)}_{\rm DRIS\mbox{-}based\;link}
        +n_{\rm w}(m), & {\mathcal H}_1,\\[0.5ex]
        n_{\rm w}(m), & {\mathcal H}_0,
    \end{array}
    \right.
    \label{CCObSigWS}
\end{equation}
where $s(m)$ denotes the $m$-th covert symbol transmitted by Alice within a channel coherence interval, and $n_{\rm w}(m)$ is additive white Gaussian noise (AWGN) with zero mean and variance $\delta_{\rm w}^2$, i.e., $n_{\rm w}(m)\sim\mathcal{CN}(0,\delta_{\rm w}^2)$. Similar to~\cite{CCRef2,CCRef3,CCRef4}, the covert symbols transmitted by Alice are modeled as independent and identically distributed (i.i.d.) complex Gaussian random variables with zero mean and variance $P_0$, i.e., $s(m)\sim\mathcal{CN}(0,P_0)$, where $P_0$ denotes the transmit power of the covert symbols.

In~\eqref{CCObSigWS}, $h^{\rm w}_{\rm d}=\widehat h^{\rm w}_{\rm d}/\mathcal{L}^{\nu^{\rm w}_{\rm d}/2}$ denotes the direct channel between Alice and Willie, where $\widehat h^{\rm w}_{\rm d}$ denotes the small-scale fading coefficient and $\mathcal{L}^{\nu^{\rm w}_{\rm d}/2}$ denotes the large-scale fading coefficient. Moreover, $h^{\rm w}_{\rm D}(m)$ denotes the cascaded DRIS-based channel between Alice and Willie at the $m$-th sampling instant. Specifically, $h^{\rm w}_{\rm D}(m)$ can be expressed as
\begin{align}
    h^{\rm w}_{\rm D}(m)
    &= {\boldsymbol g}\,{\rm diag}\!\big({\boldsymbol \varphi}(m)\big)\,{\boldsymbol h}^{\rm w}_{\rm I}
    \label{hDEx}\\
    &= \frac{\widehat{\boldsymbol g}}{\mathcal{L}^{\nu_{\rm g}/2}}
       {\rm diag}\!\big({\boldsymbol \varphi}(m)\big)
       \frac{\widehat{\boldsymbol h}^{\rm w}_{\rm I}}{\mathcal{L}^{\nu^{\rm w}_{\rm I}/2}},
    \label{hDEx1}
\end{align}
where $\mathcal{L}^{\nu_{\rm g}/2}$ and $\mathcal{L}^{\nu^{\rm w}_{\rm I}/2}$ are the large-scale fading coefficients of the Alice--DRIS and DRIS--Willie channels, respectively. Furthermore, the entries of $\widehat{\boldsymbol h}^{\rm w}_{\rm I}$ are assumed to be i.i.d. complex Gaussian random variables, i.e., $\widehat{\boldsymbol h}^{\rm w}_{\rm I}\sim\mathcal{CN}({\bf 0},{\bf I}_{N_{\rm D}})$, where ${\bf I}_{N_{\rm D}}$ is the $N_{\rm D}\times N_{\rm D}$ identity matrix.

In this work, we assume that the DRIS is deployed close to Alice to maximize its impact~\cite{IREDep}. Furthermore, the DRIS is typically equipped with a large number of reflecting elements to compensate for the severe multiplicative large-scale fading of the cascaded DRIS-based channel~\cite{MyWCMag,DIRSTWC,TWCAnti}. Consequently, the Alice--DRIS channel $\widehat{\boldsymbol g}$ is modeled using a near-field Rician model~\cite{NearfieldMoRef,NearfieldMo2}, i.e.,
\begin{equation}
    \widehat{\boldsymbol g}
    =
    \sqrt{\frac{\kappa_{\rm g}}{1+\kappa_{\rm g}}}\,\widehat{\boldsymbol g}^{\rm LOS}
    +
    \sqrt{\frac{1}{1+\kappa_{\rm g}}}\,\widehat{\boldsymbol g}^{\rm NLOS},
    \label{Ricianchan}
\end{equation}
where $\kappa_{\rm g}$ denotes the Rician factor of the Alice--DRIS channel. The entries of the non-line-of-sight (NLOS) component $\widehat{\boldsymbol g}^{\rm NLOS}$ are assumed to be i.i.d. complex Gaussian random variables, i.e., $\widehat{\boldsymbol g}^{\rm NLOS}\sim\mathcal{CN}({\bf 0},{\bf I}_{N_{\rm D}})$. The entries of the line-of-sight (LOS) component $\widehat{\boldsymbol g}^{\rm LOS}$ are given by
\begin{equation}
    \left[\widehat{\boldsymbol g}^{\rm LOS}\right]_r
    =
    e^{-j\frac{2\pi}{\lambda}(d_r-d_0)},
    \label{GLOS}
\end{equation}
where $r=1,\cdots,N_{\rm D}$, $\lambda$ is the wavelength of the transmitted covert signal, $d_r$ denotes the distance between Alice's antenna and the $r$-th DRIS element, and $d_0$ denotes the distance between Alice's antenna and the center of the DRIS, respectively.

In our considered covert communication scenario, Bob is assumed to know whether Alice is transmitting through synchronization and a pre-shared protocol/schedule. Accordingly, when Alice is silent, Bob does not attempt to receive any signal. When Alice transmits $M$ covert symbols, the received signal at Bob can be expressed as
\begin{align}
    y_{\rm b}(m)
    &= \underbrace{h^{\rm b}_{\rm d}s(m)}_{\rm Covert\;signal}
     + \underbrace{h^{\rm b}_{\rm D}(m)s(m)}_{\rm DRIS\;jamming}
     + \underbrace{n_{\rm b}(m)}_{\rm Noise}
    \label{CCObSigBS0}\\
    &= \frac{\widehat h^{\rm b}_{\rm d}s(m)}{\mathcal{L}^{\nu^{\rm b}_{\rm d}/2}}
     + \frac{\widehat{\boldsymbol g}\,{\rm diag}\!\big({\boldsymbol \varphi}(m)\big)\widehat{\boldsymbol h}^{\rm b}_{\rm I}s(m)}
            {\mathcal{L}^{\nu_{\rm g}/2}\mathcal{L}^{\nu^{\rm b}_{\rm I}/2}}
     + n_{\rm b}(m),
    \label{CCObSigBS}
\end{align}
where $m=1,2,\cdots,M$, $\mathcal{L}^{\nu^{\rm b}_{\rm d}/2}$ and $\mathcal{L}^{\nu^{\rm b}_{\rm I}/2}$ denote the large-scale fading coefficients of the Alice--Bob and DRIS--Bob channels, respectively, and $n_{\rm b}(m)$ is AWGN with zero mean and variance $\delta_{\rm b}^2$. Moreover, $\widehat h^{\rm b}_{\rm d}\sim\mathcal{CN}(0,1)$ and $\widehat{\boldsymbol h}^{\rm b}_{\rm I}\sim\mathcal{CN}({\bf 0},{\bf I}_{N_{\rm D}})$, respectively.

\subsection{Detection and Communication Performance Metrics}\label{PerformanceMer}
To evaluate Willie's monitoring performance in covert communications in the presence of the DRIS, we use the  MDR and FAR as monitoring performance metrics~\cite{CCRef3,CCRelay11,CCRelay2,C_C_RIS1,C_C_RIS4}. To evaluate the communication performance between Alice and Bob, we use the SJNR as the communication performance metric~\cite{TWCAnti}.

Based on the received samples in~\eqref{CCObSigWS}, the detection problem at Willie is formulated as a binary hypothesis test based on the test statistic
\begin{equation}
    Y_{\rm w}
    =
    \sum_{n=1}^{N}\left|y_{\rm w}(n)\right|^2,
\end{equation}
where $N$ denotes the number of samples used for detection, with $N\le M$.  Rather than performing direct threshold detection on $Y_{\rm w}$ itself, Willie makes the decision rule according to the LLR, i.e.,
\begin{equation}
    \ell\!\left(Y_{\rm w}\right)
    =
    \log f_{Y_{\rm w}\mid{\mathcal H}_1}\!\left(y\right)
    -
    \log f_{Y_{\rm w}\mid{\mathcal H}_0}\!\left(y\right).
    \label{logLLR_general}
\end{equation}
Accordingly, the decision rule at the warden Willie can be written as
\begin{equation}
    \ell\!\left(y\right)
    \mathop{\gtreqless}\limits_{{\mathcal H}_0}^{{\mathcal H}_1}
    \eta_{\rm o},
    \label{decisionrule}
\end{equation}
where $\eta_{\rm o}$ denotes the detection threshold.

Consequently, the MDR, denoted by $p_{\rm M}$, is defined as
\begin{equation}
    p_{\rm M}
    =
    \mathbb{P}\!\left(
    \left.
    \ell\!\left(y\right) < \eta_{\rm o}
    \right| {\mathcal H}_1
    \right),
    \label{PMDR}
\end{equation}
which is the probability that Willie decides ${\mathcal H}_0$ when ${\mathcal H}_1$ occurs 
and the probability that ${\mathcal H}_0$ occurs satisfies $\mathbb{P}\!\left(\left.  \cdot \right| {\mathcal H}_0\right) = \pi_1$. Similarly, the FAR, denoted by $p_{\rm F}$, is defined as
\begin{equation}
    p_{\rm F}
    =
    \mathbb{P}\!\left(
    \left.
    \ell\!\left(y\right) \ge \eta_{\rm o}
    \right| {\mathcal H}_0
    \right),
    \label{PFAR}
\end{equation}
which is the probability that Willie decides ${\mathcal H}_1$ when ${\mathcal H}_0$ occurs, and 
the probability that ${\mathcal H}_0$ occurs satisfies $\mathbb{P}\!\left(\left.  \cdot \right| {\mathcal H}_1 \right) = \pi_0=1- \pi_1$.

On the other hand, the ergodic SJNR at Bob is defined based on~\eqref{CCObSigBS} as
\begin{equation}
    \gamma_{\rm b}
    =
    \frac{
    \mathbb{E}\!\left[
    \left|
    \frac{\widehat h^{\rm b}_{\rm d}s(m)}{\mathcal{L}^{\nu^{\rm b}_{\rm d}/2}}
    \right|^2
    \right]
    }{
    \mathbb{E}\!\left[
    \left|
    \frac{\widehat{\boldsymbol g}\,{\rm diag}\!\big({\boldsymbol \varphi}(m)\big)\widehat{\boldsymbol h}^{\rm b}_{\rm I}s(m)}
    {\mathcal{L}^{\nu_{\rm g}/2}\mathcal{L}^{\nu^{\rm b}_{\rm I}/2}}
    \right|^2
    \right]
    +\delta_{\rm b}^2
    }.
    \label{SNRBob}
\end{equation}
Consequently, the achievable rate at Bob can be calculated as $R = \log_2(1+\gamma_{\rm b})$.

It follows from~\eqref{PMDR} to~\eqref{SNRBob} that the introduction of the DRIS affects both Willie's monitoring performance and Bob's ergodic SJNR. In Section~\ref{AIDec}, we further develop an unsupervised semi-parametric plug-in likelihood-ratio detector for Willie, in which a flexible normalizing flow-based model is used to learn the intractable density of $Y_{\rm w}$ under ${\mathcal H}_1$, and then quantify the impact of DRIS-induced ACA on communications between Alice and Bob.

\subsection{Preliminary: Review of Some Related Results}\label{RRs}
\newtheorem{remark}{Remark}
\begin{remark}
\label{Lemma1}
(Bayesian likelihood-ratio rule)
Let $X$ be an observation whose class-conditional densities under
${\mathcal H}_0$ and ${\mathcal H}_1$ are $f_0(x)$ and $f_1(x)$, respectively.
Let the prior probabilities of ${\mathcal H}_0$ and ${\mathcal H}_1$ be
$\pi_0$ and $\pi_1$, where $\pi_0,\pi_1>0$ and $\pi_0+\pi_1=1$.
Moreover, let $C_{ji}$ denote the decision cost incurred by deciding
${\mathcal H}_j$ when ${\mathcal H}_i$ is true, for $i,j\in\{0,1\}$.

Then, the Bayesian decision rule that minimizes the posterior expected cost
decides ${\mathcal H}_1$ if
\begin{equation}
\Lambda(x)=\frac{f_1(x)}{f_0(x)}
\mathop{\gtreqless}\limits_{{\mathcal H}_0}^{{\mathcal H}_1}
\eta_{\rm B},
\label{BayesDetector}
\end{equation}
where
\begin{equation}
\eta_{\rm B}
=
\frac{\pi_0\big(C_{10}-C_{00}\big)}
     {\pi_1\big(C_{01}-C_{11}\big)},
\label{BayesThreshold}
\end{equation}
provided that $C_{10}>C_{00}$ and $C_{01}>C_{11}$.
Equivalently, in LLR form,
\begin{equation}
\ell(x)=\log f_1(x)-\log f_0(x)
\mathop{\gtreqless}\limits_{{\mathcal H}_0}^{{\mathcal H}_1}
\log \eta_{\rm B}.
\label{BayesDetectorLog}
\end{equation}
When $\Lambda(x)=\eta_{\rm B}$, randomization may be used if necessary.

As a special case, if correct decisions incur zero cost and the two error
types have equal cost, i.e., $C_{00}=C_{11}=0$ and $C_{10}=C_{01}$, then
\begin{equation}
\eta_{\rm B}=\frac{\pi_0}{\pi_1}.
\label{BayesThresholdSpecial1}
\end{equation}
If, in addition, $\pi_0=\pi_1 = \frac{1}{2}$, then $\eta_{\rm B}=1$, and the rule reduces
to the zero-threshold test in the log domain.
\end{remark}

\begin{remark}
\label{Theorem1}
(Lindeberg-L$\acute{e}$vy Central Limit Theorem)
Suppose ${\boldsymbol x}\buildrel\Delta\over=[x_1,x_2,\cdots,x_n]$ is a vector of i.i.d. random variables with mean $\mathbb E[x_i]=\mu<\infty$ and variance ${\rm Var}(x_i)=\nu^2<\infty$ for all $i=1,\cdots,n$. Then, according to the Lindeberg-L$\acute{e}$vy central limit theorem,
\begin{equation}
\sqrt{n}\left(\overline X-\mu\right)
=
\frac{\sum_{i=1}^{n}x_i}{\sqrt n}-\sqrt n\,\mu
\mathop{\to}\limits^{\rm d}
\mathcal{N}(0,\nu^2),
\qquad n\to\infty,
\label{CLTeq}
\end{equation}
where $\overline X=\frac{\sum_{i=1}^{n}x_i}{n}$.
\end{remark}

\section{Simultaneously Exposing and Jamming Covert Communications}\label{AIDec}
In Section~\ref{WillieDec}, we first describe the challenges that the DRIS introduces for Willie's detection, under which the PDF of the test statistic at Willie under the transmission hypothesis ${\mathcal{H}}_1$ becomes analytically intractable. 
In Section~\ref{WillieDecAI}, we then develop an unsupervised semi-parametric plug-in likelihood-ratio detector, which retains the parametric Gamma reference model under ${\mathcal H}_0$ without requiring prior knowledge of the noise power, and learns from unlabeled data a one-dimensional monotone normalizing flow model for the analytically intractable density under ${\mathcal H}_1$. 
In Section~\ref{Bobcommun}, we derive the statistical characteristics of the cascaded DRIS channel between Alice and Bob, and then conduct an asymptotic analysis of the SJNR at Bob to quantify the communication performance. Some interesting properties of DRIS on covert communications are also derived based on the asymptotic analysis.

\subsection{Impact of DRIS on Willie's Detection}\label{WillieDec}

Recalling Willie's decision rule in~\eqref{decisionrule}, this is a binary hypothesis testing problem. 
When Alice and Bob are silent, i.e., under the silence hypothesis ${\mathcal H}_0$, the observation $Y_{\rm w}$ reduces to
\begin{equation}
    Y_{\rm w}=\sum_{n=1}^{N}\left|n_{\rm w}(n)\right|^2.
    \label{ObserveH0}
\end{equation}

Since $\{n_{\rm w}(n)\}_{n=1}^{N}$ are AWGN samples, $Y_{\rm w}$ follows a Gamma distribution with shape parameter $N$ and scale parameter $\delta_{\rm w}^2$, i.e.,
\begin{equation}
    Y_{\rm w}\sim {\rm Gamma}\!\left(N,\delta_{\rm w}^2\right).
\end{equation}
Accordingly, the PDF of $Y_{\rm w}$ under ${\mathcal H}_0$ is
\begin{equation}
    f_{Y_{\rm w}\mid {\mathcal H}_0}(y)
    =
    \frac{1}{\Gamma(N)\left(\delta_{\rm w}^2\right)^N}y^{N-1}e^{-y/\delta_{\rm w}^2},
    \qquad y\ge 0,
    \label{ObserveH0PDF}
\end{equation}
where $\Gamma(\cdot)$ denotes the Gamma function. Since $N$ is a positive integer, $\Gamma(N)=(N-1)!$.

If the PDF of $Y_{\rm w}$ under ${\mathcal H}_1$ were available, the exact NP detector could be constructed accordingly. Before discussing this issue, we first derive the asymptotic distribution of the cascaded DRIS-based channel $h_{\rm D}^{\rm w}(t)$ in~\eqref{hDEx}, as stated in Proposition~\ref{Proposition1}.

\newtheorem{proposition}{Proposition}
\begin{proposition}
\label{Proposition1}
The random and time-varying DRIS-based term $h_{\rm D}^{\rm w}(t)$ converges in distribution to a zero-mean complex Gaussian random variable as $N_{\rm D}\to\infty$, i.e.,
\begin{equation}
h_{\rm D}^{\rm w}(t)
=
\frac{\widehat{\boldsymbol g}\,
{\rm diag}\!\big(\boldsymbol \varphi(t)\big)\,
\widehat{\boldsymbol h}_{\rm I}^{\rm w}}
{\mathcal{L}^{\nu_{\rm g}/2}\mathcal{L}^{\nu_{\rm I}^{\rm w}/2}}
\mathop{\longrightarrow}\limits^{\rm d}
\mathcal{CN}\!\left(
0,\,
\frac{N_{\rm D}\overline{\alpha}}
{\mathcal{L}^{\nu_{\rm g}}\mathcal{L}^{\nu_{\rm I}^{\rm w}}}
\right).
\label{HDSta}
\end{equation}
In~\eqref{HDSta}, $\overline{\alpha}=\sum_{i=1}^{2^b} p_i \alpha_i^2$
where $p_i$ denotes the probability that the phase shift of the $r$-th DRIS element takes the $i$-th value in $\Omega$, i.e., $p_i=\mathbb{P}\!\left(\varphi_r(t)=\phi_i\right), \forall r$.
\end{proposition}

\begin{IEEEproof}
    See Appendix~\ref{AppendixA}.
\end{IEEEproof}

Although the DRIS coefficients are time-varying, throughout this work they are generated according to a time-invariant random law. Specifically, each DRIS element takes values from $\Omega$ with fixed probabilities $p_i$ i.e., the distribution of $\varphi_r(t)$ does not drift with $t$. Under this assumption, the test statistic received at Willie, i.e., $Y_{\rm w}$ admits a time-invariant marginal distribution under ${\mathcal H}_1$. Accordingly, $f_{Y_{\rm w}\mid {\mathcal H}_1}(y)$ can be treated as stationary over the considered operating interval, which justifies learning a single time-invariant density model from samples collected across sensing intervals.

Unfortunately, according to Proposition~\ref{Proposition1}, the deployment of the DRIS prevents Willie from deriving an explicit expression for the PDF of $Y_{\rm w}$ under ${\mathcal{H}}_1$, i.e., $f_{Y_{\rm w}|\mathcal H_1}(y)$. It is worth noting that the DRIS jamming term in~\eqref{CCObSigWS}, i.e., $Z_{\rm w} = h_{\rm D}^{\rm w}(m)s(m)$, can be viewed as a new random variable arising from the product of two independent complex Gaussian random variables. However, the product of two independent complex Gaussian random variables is not itself Gaussian~\cite{CNN}.
Given two independent complex random variables $X$ and $Y$ with PDFs
$f_X(x)$ and $f_Y(y)$, respectively, the PDF of $Z=XY$ can be computed as
\begin{equation}
    {f_Z}\!\left( z \right) = \int_{ \mathbb{C} }  {{f_X}\left( x \right)} {f_Y}\!\left( {\frac{z}{x}} \right)\frac{1}{{\left| x \right|^2}}d^2x.
    \label{XYPDF}
\end{equation}
Consequently, the PDF of $Z_{\rm w}$ can be calculated as
\begin{equation}
\begin{split}
f_{Z_{\rm w}}\!\left( z \right) &= \frac{1}{\pi^2 P_0 \frac{ N_{\rm D}\overline \alpha}{ \mathcal{L}^{\nu_{{\rm{g}}}} \mathcal{L}^{\nu^{\rm w}_{{\rm{I}}}} }} 
\\
&\quad \times \int_{ \mathbb{C} }  \exp\!\!\left\{ -\frac{\left|x\right|^2}{ \frac{ N_{\rm D}\overline \alpha}{ \mathcal{L}^{\nu_{{\rm{g}}}} \mathcal{L}^{\nu^{\rm w}_{{\rm{I}}}} } } \right\} 
\exp\!\!\left\{ -\frac{ \left| \frac{z}{x} \right|^2 }{ P_0 } \right\}
\frac{1}{ \left| x \right|^2 } d^2x.
\label{SpeXYPDF}
\end{split}
\end{equation}

Moreover, in~\eqref{CCObSigWS} under ${\mathcal{H}}_1$, the sum of the direct-link covert term and the noise can be treated as a new complex Gaussian random variable, i.e., $N_{\rm w} = {{ {{h^{\rm{w}}_{{\rm{d}}}}s(m)} }} + {n_{{\rm w}}(m)} \sim {\mathcal{CN}}\!\left(0, {\left|h_{\rm d}^{\rm w}\right|^2 P_0 + \delta^2_{\rm w}}\right)$.
Note that, given two independent complex random variables $X$ and $Y$ with PDFs
$f_X(x)$ and $f_Y(y)$, respectively, the PDF of $Z=X+Y$ can be calculated as
\begin{equation}
    {f_Z}\!\left( z \right) = \int_{ \mathbb{C} }  {{f_X}\!\left( x \right)} {f_Y}\!\left( z-x \right) d^2x.
    \label{XplusYPDF}
\end{equation}
Theoretically, the $m$-th received signal at Willie under ${\mathcal{H}}_1$, i.e., $y_{\rm w}(m) = Z_{\rm w} + N_{\rm w}$, can be regarded as a random variable whose PDF can be calculated using~\eqref{XplusYPDF}. Unfortunately, based on~\eqref{SpeXYPDF} and~\eqref{XplusYPDF}, it is difficult to derive an explicit expression for the PDF of $y_{\rm w}(m)$ because of the mathematical intractability introduced by the product and sum of complex random variables. As a result, it is also difficult to derive a closed-form expression for the PDF of $Y_{\rm w}$ under ${\mathcal{H}}_1$. Therefore, in the covert communication system in the presence of a DRIS, the exact NP detector, and hence the corresponding decision threshold $\eta_{\rm o}$, cannot be obtained in closed form.

\subsection{Unsupervised Semi-parametric Plug-in Likelihood-Ratio Detection}\label{WillieDecAI}
As discussed in Section~\ref{WillieDec}, the difficulty of detection at Willie is twofold. Under the silence hypothesis ${\mathcal H}_0$, although the PDF of the test statistic belongs to a Gamma family as implied by~\eqref{ObserveH0PDF}, Willie does not have access to the noise variance $\delta^{2}_{\rm w}$, and hence the corresponding scale parameter is unknown a priori. Under the transmission hypothesis ${\mathcal H}_1$, the distribution of the test statistic $Y_{\rm w}$ does not admit a tractable closed-form expression. Moreover, Willie is assumed to have no access to  labeled training samples. Motivated by these facts, we construct an unsupervised semi-parametric plug-in likelihood-ratio detector for Willie to monitor whether Alice and Bob are transmitting. The detector is semi-parametric in the sense that it retains a parametric Gamma reference model under ${\mathcal H}_0$ while jointly estimating its effective scale from unlabeled data, and models the intractable density under ${\mathcal H}_1$ using a flexible one-dimensional monotone normalizing flow model.

Let ${\mathcal Y}=\{Y^t_{\rm w}\}_{t=1}^{T_{\rm w}}$ denote a collection of unlabeled realizations of the test statistic $Y_{\rm w}$. Since the DRIS coefficients are generated according to a time-invariant random law, the distribution of $Y_{\rm w}$ under each hypothesis is stationary over the considered operating interval. Therefore, the samples in ${\mathcal Y}$ can be treated as i.i.d. draws from a common mixture distribution. Our objective is to recover, from ${\mathcal Y}$ alone, both the density under ${\mathcal H}_1$ and the effective scale parameter of the Gamma reference density under ${\mathcal H}_0$. To this end, we model the unlabeled observations by
\begin{equation}
f_{Y_{\rm w}}(y)
=
\pi_0 f_{Y_{\rm w}\mid{\mathcal H}_0}(y;\delta_{\rm w}^2)
+
\pi_1 f_{Y_{\rm w}\mid{\mathcal H}_1}(y;\boldsymbol{\psi}),
\label{SemiMix}
\end{equation}
where $y\ge 0$, $\pi_0,\pi_1\in(0,1)$ satisfy $\pi_0+\pi_1=1$, $\delta_{\rm w}^2>0$ denotes the scale parameter of the Gamma reference density under ${\mathcal H}_0$, and $\boldsymbol{\psi}$ collects the parameters of the flow model under ${\mathcal H}_1$.

The null-hypothesis component (i.e., the silence hypothesis ${\mathcal H}_0$) in~\eqref{SemiMix} retains the Gamma form implied by~\eqref{ObserveH0PDF}, but its scale parameter $\delta^2_{\rm w}$ is not assumed known a priori and is jointly estimated from unlabeled data. By contrast, the density under the transmission hypothesis ${\mathcal H}_1$ is learned from data through a one-dimensional monotone normalizing flow model. 
Specifically, let ${\mathcal T}(\cdot;\boldsymbol{\psi})$ be a differentiable and strictly monotone transformation that maps $y$ to a latent variable $z$ with a simple reference density $p_Z(z)$, i.e.,
\begin{equation}
z={\mathcal T}(y;\boldsymbol{\psi}) \sim p_Z(z),
\label{ScalarFlow}
\end{equation}
where $p_Z(z)$ is chosen as the standard Gaussian density. By the change-of-variables formula, the resulting density under ${\mathcal H}_1$ is
\begin{equation}
f_{Y_{\rm w}\mid{\mathcal H}_1}(y;\boldsymbol{\psi})
=
p_Z\!\big({\mathcal T}(y;\boldsymbol{\psi})\big)
\left|
\frac{d{\mathcal T}(y;\boldsymbol{\psi})}{dy}
\right|.
\label{FlowDensity}
\end{equation}

Based on~\eqref{ScalarFlow} and \eqref{FlowDensity}, the unlabeled likelihood detector design can be stated as a constrained semi-parametric estimation problem as follows:
\begin{subequations}
\begin{align}
\max_{\delta_{\rm w}^2,\boldsymbol{\psi},\{r_t\}}
&
\sum_{t=1}^{T_{\rm w}}\! 
\Big[\! 
(1\!-\!r_t)\log\! f_{Y_{\rm w}\mid{\mathcal H}_0}(y_t;\delta_{\rm w}^2)
\!+\!
r_t \log \!f_{Y_{\rm w}\mid{\mathcal H}_1}(y_t;\boldsymbol{\psi})\! 
\Big]
\label{SemiObjective}\\
\text{s.t.}\quad
&
0\le r_t \le 1, \nonumber\\
&
\delta_{\rm w}^2>0, \nonumber\\
&
\frac{1}{T_{\rm w}}\sum_{t=1}^{T_{\rm w}} r_t = \pi_1,
\label{SemiConstraint}
\end{align}
\end{subequations}
where $r_t$ is the posterior responsibility associated with ${\mathcal H}_1$ for the sample $y_t$ and $t=1,\ldots,T_{\rm w}$. In other words, $r_t$ acts as a soft hypothesis label inferred directly from the unlabeled data.

To initialize the estimation of the Gamma scale parameter $\delta_{\rm w}^2$, we exploit the structural prior that the lower tail of ${\mathcal Y}$ is dominated by samples that are more likely to arise from ${\mathcal H}_0$. Let ${\mathcal Y}_{\rm L}\subset{\mathcal Y}$ denote the subset formed by the smallest $q$ fraction of the observations, where $q\in(0,1)$ is a design parameter. Denote by $\overline y_{\rm L}$ the sample mean over ${\mathcal Y}_{\rm L}$. Matching this empirical lower-tail mean to the corresponding truncated mean of a unit-scale Gamma random variable yields an anchor estimate
\begin{equation}
\widehat{\delta}_{{\rm w},{\rm anc}}^{2}
=
\frac{\overline y_{\rm L}}{m_N(q)},
\label{AnchorTheta}
\end{equation}
where $m_N(q)$ denotes the mean of a ${\rm Gamma}(N,1)$ random variable conditioned on belonging to its lower $q$-quantile. The quantity $\widehat{\delta}_{{\rm w},{\rm anc}}^{2}$ is used only to stabilize the unsupervised iterations and is not itself the final estimate.

The anchor estimate in~\eqref{AnchorTheta} is next used to bootstrap the one-dimensional monotone normalizing flow model. Specifically, based on the reference density under ${\mathcal H}_0$, Willie first assigns larger preliminary ${\mathcal H}_1$ weights to samples that are less likely to be explained by the Gamma reference model. Let $\{r_t^{(0)}\}_{t=1}^{T_{\rm w}}$ denote the resulting initial responsibilities, where $r_t^{(0)}\in(0,1)$ and larger observations are assigned larger initial probabilities of belonging to ${\mathcal H}_1$. Using these initial responsibilities, Willie obtains an initial flow estimate by solving the weighted maximum-likelihood problem
\begin{equation}
\widehat{\boldsymbol{\psi}}^{(0)}
=
\arg\max_{\boldsymbol{\psi}}
\sum_{t=1}^{T_{\rm w}}
r_t^{(0)}
\log f_{Y_{\rm w}\mid{\mathcal H}_1}(y_t;\boldsymbol{\psi}).
\label{BootstrapFlow}
\end{equation}
Therefore, the subsequent alternating refinement is initialized by $\delta_{\rm w}^2=\widehat{\delta}_{{\rm w},{\rm anc}}^{2}$ and $\boldsymbol{\psi}=\widehat{\boldsymbol{\psi}}^{(0)}$.

Starting from this bootstrap initialization, the parameters are refined by alternating updates. Given the current parameter values, the posterior responsibility of ${\mathcal H}_1$ for the sample $y_t$ is computed as
\begin{equation}
r_t
=
\frac{
\pi_1 f_{Y_{\rm w}\mid{\mathcal H}_1}(y_t;\boldsymbol{\psi})
}{
\pi_0 f_{Y_{\rm w}\mid{\mathcal H}_0}(y_t;\delta_{\rm w}^2)
+
\pi_1 f_{Y_{\rm w}\mid{\mathcal H}_1}(y_t;\boldsymbol{\psi})
},
\label{PosteriorRespRaw}
\end{equation}
for $t=1,\ldots,T_{\rm w}$. In practice, to maintain consistency with the prescribed prior mixture proportion $\pi_1$, a scalar offset can be introduced in the LLR before the logistic mapping so that the empirical mean of the updated responsibilities satisfies the constraint in~\eqref{SemiConstraint}. This leads to the calibrated form
\begin{equation}
r_t
=
\left[
1+\exp\!\left\{
-
\log\frac{f_{Y_{\rm w}\mid{\mathcal H}_1}(y_t;\boldsymbol{\psi})}
{f_{Y_{\rm w}\mid{\mathcal H}_0}(y_t;\delta_{\rm w}^2)}
-\log\frac{\pi_1}{\pi_0}
-\Delta
\right\} \!
\right]^{-1},
\label{PosteriorResp}
\end{equation}
where $\Delta$ is chosen such that $\frac{1}{T_{\rm w}}\sum_{t=1}^{T_{\rm w}}r_t=\pi_1$.

Given $\{r_t\}_{t=1}^{T_{\rm w}}$, the maximum-likelihood update of the Gamma scale parameter under ${\mathcal H}_0$ is
\begin{equation}
\widehat{\delta}_{{\rm w},{\rm ML}}^{2}
=
\frac{\sum_{t=1}^{T_{\rm w}} (1-r_t)y_t}
{N\sum_{t=1}^{T_{\rm w}} (1-r_t)}.
\label{ThetaML}
\end{equation}
To suppress the drift caused by inevitable soft-label contamination in the fully unsupervised setting, we adopt a shrinkage update
\begin{equation}
\widehat{\delta}_{\rm w}^{2}
=
(1-\lambda)\widehat{\delta}_{{\rm w},{\rm ML}}^{2}
+
\lambda \widehat{\delta}_{{\rm w},{\rm anc}}^{2},
\label{ThetaShrink}
\end{equation}
where $\lambda$ is a regularization coefficient and $0\le \lambda < 1$. Thus, the updated Gamma scale parameter $\widehat{\delta}_{\rm w}^{2}$ balances the current soft-label estimate and the lower-tail anchor, which improves numerical stability without imposing a noise prior.

\newcounter{myalgo}
\renewcommand{\themyalgo}{\arabic{myalgo}}

\begin{figure}[!t]
\refstepcounter{myalgo}
\hrule
\vspace{0.5ex}
\noindent\textbf{Algorithm \themyalgo} Unsupervised Semi-parametric Plug-in LLR Detector
\label{TabCode}
\vspace{0.5ex}
\hrule
\vspace{0.6ex}

\begin{minipage}{\columnwidth}
\begin{algorithmic}[1]
\Statex \textbf{Input:} Unlabeled energy samples $\mathcal{Y}=\{y_t\}_{t=1}^{T_{\rm w}}$, number of samples $N$, prior probabilities $(\pi_0,\pi_1)$, lower-tail fraction $q$, shrinkage factor $\lambda$, and maximum number of alternating iterations $I_{\max}$.

\Statex \textbf{Initialize:}
\Statex \hspace{\algorithmicindent} Build the semi-parametric mixture model in~\eqref{SemiMix}.
\Statex \hspace{\algorithmicindent} Use~\eqref{ObserveH0PDF} for the ${\mathcal H}_0$ density, and use~\eqref{ScalarFlow}--\eqref{FlowDensity} for the estimated ${\mathcal H}_1$ density.
\Statex \hspace{\algorithmicindent} Compute the anchor scale $\widehat{\delta}_{{\rm w},{\rm anc}}^{2}$ by~\eqref{AnchorTheta}, and set $\delta_{\rm w}^{2}=\widehat{\delta}_{{\rm w},{\rm anc}}^{2}$.
\Statex \hspace{\algorithmicindent} Build the initial soft responsibilities $\{r_t^{(0)}\}_{t=1}^{T_{\rm w}}$.
\Statex \hspace{\algorithmicindent} Compute the initial flow parameters $\widehat{\boldsymbol{\psi}}^{(0)}$ by~\eqref{BootstrapFlow}, and set $\boldsymbol{\psi}=\widehat{\boldsymbol{\psi}}^{(0)}$.

\For{$i=1,2,\ldots,I_{\max}$}
    \State Update the soft responsibilities $\{r_t\}_{t=1}^{T_{\rm w}}$ by \eqref{PosteriorRespRaw} or \eqref{PosteriorResp}, enforcing the mixture constraint in~\eqref{SemiConstraint}.
    \State Update the Gamma scale estimate $\widehat{\delta}_{{\rm w},{\rm ML}}^{2}$ by~\eqref{ThetaML}.
    \State Update the scale $\widehat{\delta}_{\rm w}^{2}$ by~\eqref{ThetaShrink}.
    \State Update the flow parameters $\widehat{\boldsymbol{\psi}}$ by~\eqref{MLEtheta} or~\eqref{WeightedNLL}.
    \State Set $\delta_{\rm w}^{2}=\widehat{\delta}_{\rm w}^{2}$ and $\boldsymbol{\psi}=\widehat{\boldsymbol{\psi}}$.
\EndFor

\State Compute the plug-in LLR $\hat{\ell}(y)$ by~\eqref{logLLR}.
\State Perform the detection at Willie by~\eqref{OptimaletaCom}. 

\end{algorithmic}
\end{minipage}

\vspace{0.4ex}
\hrule
\end{figure}

The flow parameters are then updated by weighted maximum likelihood. Specifically, for the current responsibilities, Willie solves
\begin{equation}
\widehat{\boldsymbol{\psi}}
=
\arg\max_{\boldsymbol{\psi}}
\sum_{t=1}^{T_{\rm w}}
r_t\,
\log f_{Y_{\rm w}\mid{\mathcal H}_1}(y_t;\boldsymbol{\psi}).
\label{MLEtheta}
\end{equation}
Equivalently, the optimization can be written as the weighted negative log-likelihood minimization
\begin{equation}
\mathcal{L}(\boldsymbol{\psi})
=
-
\frac{
\sum_{t=1}^{T_{\rm w}} r_t
\log f_{Y_{\rm w}\mid{\mathcal H}_1}(y_t;\boldsymbol{\psi})
}{
\sum_{t=1}^{T_{\rm w}} r_t
}.
\label{WeightedNLL}
\end{equation}
In implementation, the samples used in the flow update may be standardized to improve numerical conditioning, while the corresponding Jacobian correction is retained so that the density estimate remains valid on the original scale.

After the alternating updates converge, Willie obtains the plug-in likelihood ratio
\begin{equation}
\hat{\Lambda}(y)
=
\frac{
f_{Y_{\rm w}\mid{\mathcal H}_1}(y;\widehat{\boldsymbol{\psi}})
}{
f_{Y_{\rm w}\mid{\mathcal H}_0}(y;\widehat{\delta}_{\rm w}^{2})
},
\label{LLR}
\end{equation}
and the corresponding plug-in LLR
\begin{equation}
\hat{\ell}(y)
=
\log f_{Y_{\rm w}\mid{\mathcal H}_1}(y;\widehat{\boldsymbol{\psi}})
-
\log f_{Y_{\rm w}\mid{\mathcal H}_0}(y;\widehat{\delta}_{\rm w}^{2}),
\label{logLLR}
\end{equation}
where $y\ge 0$. Consequently, the final decision rule is
\begin{equation}
\hat{\ell}(y)
\mathop{\gtreqless}\limits_{{\mathcal H}_0}^{{\mathcal H}_1}
{\widehat \eta _{\rm{o}}},
\label{OptimaletaCom}
\end{equation}
where, in the general Bayesian formulation, ${\widehat \eta _{\rm{o}}}$ depends on the prior probabilities and decision costs. Assuming equal prior and error costs under ${{\mathcal H}_0}$ and ${{\mathcal H}_1}$, the rule reduces to the zero-threshold test
\begin{equation}
\hat{\ell}(y)
\mathop{\gtreqless}\limits_{{\mathcal H}_0}^{{\mathcal H}_1}
0.
\label{ZeroThreshold}
\end{equation}

It is important to note that the detector in~\eqref{OptimaletaCom} is a plug-in LLR detector rather than the NP detector, since the density under ${\mathcal H}_1$ is learned from unlabeled data and the Gamma scale parameter under ${{\mathcal H}_0}$ is jointly estimated rather than assumed known a priori. Nevertheless, this construction preserves the exact parametric structure of the Gamma reference density under ${\mathcal H}_0$, while allowing the intractable density under ${\mathcal H}_1$ to be represented by a flexible one-dimensional monotone normalizing flow model. The overall procedure is illustrated in Table~\ref{TabCode}.
 
\subsection{Impact of DRIS on Communications Between Alice and Bob}\label{Bobcommun}
According to the SJNR defined in~\eqref{SNRBob}, the implementation of the DRIS also affects the communication performance between Alice and Bob.
To quantify the impact of the DRIS on covert communications between Alice and Bob, we first derive the following statistical characteristic of the cascaded DRIS-jammed channel $h_{\rm D}^{\rm b}(t)$ in Proposition~\ref{Proposition2}.

\begin{proposition}
    \label{Proposition2}
    The random and time-varying DRIS-based term
    $h_{\rm D}^{\rm b}(t)$ converges in distribution to a complex Gaussian discussion as $N_{\rm D} \to \infty$, i.e.,
    \begin{equation}
       h_{\rm D}^{\rm b}(t)=  { {{\frac{{{\widehat{{\boldsymbol g }}}}
     {{\rm{diag}}\!\left({\boldsymbol \varphi }(t) \right)} {\widehat{\boldsymbol{h}}^{{\rm b}}_{{\rm I}} } }{{  {\mathcal{L}^{\frac{{{\nu_{{\rm{g}}}}}}{2}}} {\mathcal{L}^{\frac{{{\nu^{\rm{b}}_{{\rm{I}}}}}}{2}}} } }}} }
         \mathop  \to \limits^{\rm{d}} \mathcal{CN}\!\left( {0,  {\frac{ N_{\rm D}\overline \alpha}{{ {\mathcal{L}^{{{{\nu _{{\rm{g}}}}}}}} {\mathcal{L}^{{{{\nu^{\rm b}_{{\rm{I}} }}}}}} } }} } \right).
        \label{HDStaBob}
    \end{equation}
\end{proposition}

\begin{IEEEproof}
    See Appendix~\ref{AppendixB}.
\end{IEEEproof}

Conditioned on the fact that the covert symbol $s(m)$ is independent of the wireless channels, the SJNR in~\eqref{SNRBob} can be reduced to
 \begin{equation}
    {\gamma _{\rm{b}}} = \frac{\mathbb{E} \!\!\left[\left|{ {{{ h}^{\rm{b}}_{{\rm{d}}}} } } { { s(m)} }\right|^2\right]}
    {\mathbb{E} \!\!\left[\left| { {{{ h}^{\rm{b}}_{{\rm{D}}}(m)}s(m)} }\right|^2 \right]+ \delta^2_{\rm b}}.
     \label{SNRBobReduce}
\end{equation}
Substituting~\eqref{HDStaBob} into~\eqref{SNRBobReduce}, we have
 \begin{equation}
    {\gamma _{\rm{b}}} \mathop  \to \limits^{\rm{d}}  \frac{\frac{P_0}{{\mathcal{L}^{{{{\nu^{\rm b}_{{\rm{d}} }}}}}}} }
    { {\frac{ P_0 N_{\rm D}\overline \alpha}{{ {\mathcal{L}^{{{{\nu _{{\rm{g}}}}}}}} {\mathcal{L}^{{{{\nu^{\rm b}_{{\rm{I}} }}}}}} } }}+ \delta^2_{\rm b}}, \; {\rm{as}}\; N_{\rm D} \to \infty.
     \label{SNRBobReduceAr}
\end{equation}

In traditional covert communications, while higher transmit power improves the communication performance between Alice and Bob, it also increases the probability of detection by the warden Willie. However, in the cover communication system in the presence of a DRIS, it can be seen from~\eqref{SNRBobReduceAr} that increasing transmit power $P_0$ not only results in a higher detection probability at the warden, but also amplifies the DRIS-induced ACA, i.e., ${\frac{ P_0 N_{\rm D}\overline \alpha}{{ {\mathcal{L}^{{{{\nu _{{\rm{g}}}}}}}} {\mathcal{L}^{{{{\nu^{\rm b}_{{\rm{I}} }}}}}} } }}$, which degrades the communication performance between Alice and Bob rather than enhancing it.
The DRIS not only jams the covert transmissions between Alice and Bob, but also decreases the error probabilities of Willie's detections, without either Alice--Bob channel knowledge or additional jamming power.

Based on the theoretical SJNR in~\eqref{SNRBobReduceAr}, the effect of the DRIS quantization resolution is seen to enter primarily through
$\overline{\alpha}=\sum_{i=1}^{2^b} p_i \alpha_i^2$, namely, the average squared reflection amplitude induced by the random DRIS coefficients. Hence, for a special configuration and hardware response, increasing the DRIS phase quantization resolution does not necessarily lead to a substantial additional degradation in Bob's SJNR, because the resulting change in $\overline{\alpha}$ is often limited. This indicates that the performance gain obtained by using finer DRIS phase quantization is generally marginal. Consequently, a 1-bit phase quantization DRIS can already produce sufficiently strong DRIS-induced ACA to degrade  communications between Alice and Bob. Combined with~\eqref{SNRBobReduceAr}, this also shows that increasing Alice's transmit power does not significantly improve Bob's SJNR in the presence of the DRIS. Instead, it amplifies the DRIS-induced ACA and increases Alice's risk of detection by Willie.

\section{Simulation Results and Discussion}\label{Simu}
In this section, we perform Monte Carlo simulations to evaluate the impact of the DRIS on covert communications and to validate our theoretical analysis and the proposed unsupervised semi-parametric plug-in likelihood-ratio detector in Section~\ref{AIDec}. Unless otherwise stated, the default parameters are as follows. Alice, equipped with a single antenna, is located at (0 m, 0 m, 5 m), while Bob, also equipped with a single antenna, is uniformly distributed within an annular region centered at (0 m, 140 m, 0 m) with inner and outer radii of 10 m and 20 m, respectively.
The warden Willie is positioned at (0 m, 100 m, 0 m) to monitor potential covert transmissions.
A DRIS consisting of 2048 reflective elements ($N_{{\rm D},h} = 64, N_{{\rm D},v} = 32$) is deployed at ($-d _ { A D }$ m, 0 m, 5 m), with the distance between Alice and the DRIS center set to $d _ { A D } = 1.5$ m.
The DRIS employs 1-bit quantized reflection coefficients, with phase shifts chosen equiprobably from $\Psi = \{\frac{\pi}{9},\frac{7\pi}{6}\}$ and thier corresponding amplitudes determined by the hardware mapping $\Omega  = {\cal F}\!\left({\Theta}\right) = \{0.8,1\}$~\cite{CRIS1}. This configuration results in $\overline \alpha $ = 0.82, consistent with Proposition~\ref{Proposition1}.

The received signals at Willie are processed using the detection rule in~\eqref{decisionrule}, accumulating $N=5$ samples within each channel coherence interval. Large-scale fading parameters follow 3GPP standards~\cite{3GPP} as specified in Table~\ref{tab1}. The AWGN variance is $\sigma^2_{\rm c} = -170 + 10\log_{10}\left(BW\right)$ dBm with $BW = 180$ kHz. 
In the proposed unsupervised detector, we employ the first 16384 test statistics (6.25\%) in the unlabeled data set for training, and the remaining 245760 test statistics (93.75\%) are used for performance evaluation. 
Furthermore, the unlabeled test statistics are modeled by the semi-parametric mixture in Section~\ref{WillieDecAI}: the ${\mathcal H}_0$ component retains the Gamma form with fixed shape parameter $N=5$ and unknown scale parameter $\delta^2_{\rm w}$, while the ${\mathcal H}_1$ component is represented by a flexible one-dimensional monotone normalizing flow model. 

\begin{table}
    \footnotesize
    \centering
    \caption{Wireless Channel Simulation Parameters}
    \label{tab1}
    \begin{threeparttable}
    \begin{tabular}{ c|c }
    \hline
    Large-scale Parameter        &Value\\
    \hline
    LoS fading       & $ 35.6 + 22{\log _{10}}({d}) $ (dB) \\
    \hline
    NLoS fading      &$32.6+36.7{\log _{10}}({d})$ \\
    \hline
    \end{tabular}
    \end{threeparttable}
\end{table} 

In the simulations, the prior probabilities of ${\mathcal H}_0$ and ${\mathcal H}_1$ occurring are set to $(\pi_0,\pi_1)=(0.5,0.5)$, consistent with the equal-prior Bayesian decision rule. For initialization, the lower-tail fraction is set to $q=0.08$ to compute the anchor scale $\widehat{\delta}_{{\rm w},{\rm anc}}^{2}$, and the initial soft responsibilities are then constructed from the Gamma-tail scores under ${\mathcal H}_0$. The ${\mathcal H}_1$ density is implemented by a one-dimensional monotone normalizing flow composed of four monotone layers, each using eight nonlinear terms. The flow model is trained using Adam with learning rate $2\times 10^{-4}$, where the bootstrap initialization stage uses 180 epochs and each subsequent alternating update uses 90 epochs. The alternating refinement is run for $I_{\max}=4$ iterations, and the Gamma scale update under ${\mathcal H}_0$ uses the shrinkage factor $\lambda=0.4$.  

Fig.~\ref{fig3} illustrates Willie's DEP on the left y-axis 
and Bob's achievable rate (bits/symbol) on the right y-axis versus Alice's transmit power.
Specifically, the DEP curves include: 1) the supervised plug-in likelihood-ratio detector without a DRIS (Supervised W/O DRIS), 2) the proposed unsupervised semi-parametric plug-in likelihood-ratio detector without a DRIS (Proposed W/O DRIS), 3) the benchmark scheme without a DRIS in~\cite{huang2025simultaneously} (DEP W/O DRIS~\cite{huang2025simultaneously}), 4) the supervised plug-in likelihood-ratio detector with a DRIS (Supervised W/ DRIS), 5) the proposed unsupervised semi-parametric plug-in likelihood-ratio detector with a DRIS (Proposed W/ DRIS), and 6) the benchmark scheme with a DRIS in~\cite{huang2025simultaneously} (DEP W/ DRIS~\cite{huang2025simultaneously}). 
The achievable-rate curves include: 1) the rate without a DRIS (Rate W/O DRIS), 2) the rate with a DRIS (Rate W/ DRIS), 3) the theoretical rate computed from~\eqref{SNRBobReduceAr} (Theoretical in~\eqref{SNRBobReduceAr}), 4) the rate under an active jammer (AJ) with $-7$ dBm jamming power (AJ @ $-7$ dBm), and 5) the rate under an AJ with $3$ dBm jamming power (AJ @ $+3$ dBm).

\begin{figure}[!t]
\centering
\includegraphics[scale=0.61]{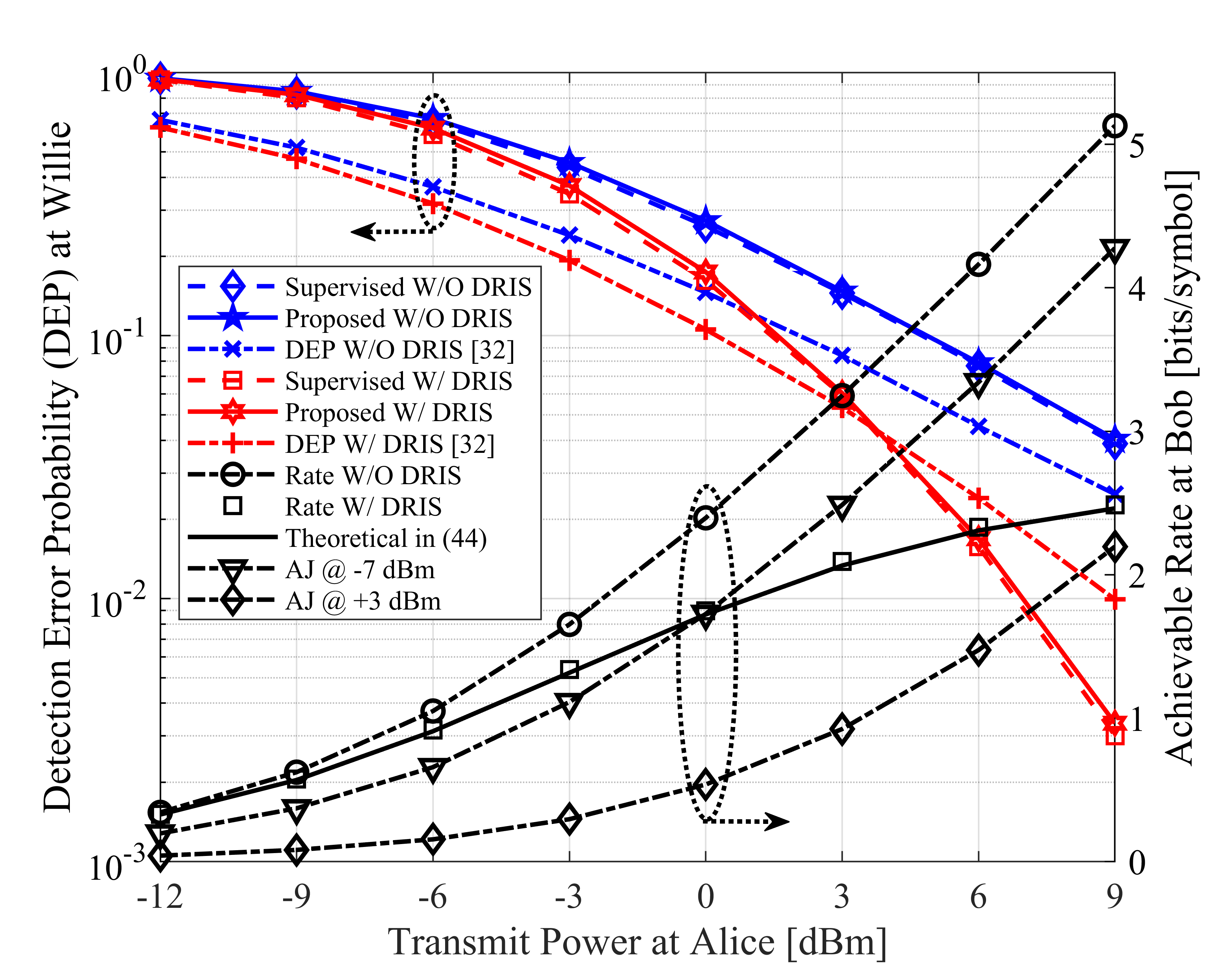}
\caption{Detection error probability (DEP) at Willie (left y-axis) and achievable rate at Bob (right y-axis) vs. transmit power.}
\label{fig3}
\end{figure}

From Fig.~\ref{fig3}, first, the DEP at Willie decreases monotonically with Alice's transmit power for all considered schemes. 
Moreover, for a given transmit power, the DEP curves with a DRIS are consistently below their counterparts without a DRIS, which confirms that the DRIS improves Willie's monitoring performance. 
More importantly, the DEP achieved by the proposed unsupervised detector closely tracks that of its supervised counterpart in both scenarios, especially when the DRIS is present. 
This result indicates that the proposed detector can accurately capture the analytically intractable density under ${\mathcal H}_1$ from unlabeled data while jointly estimating the effective Gamma scale $\widehat \delta _{\rm{w}}^2$ under ${\mathcal H}_0$, without relying on prior knowledge of the noise variance. 
The reference DEP curves from~\cite{huang2025simultaneously} exhibit the same overall trend, further corroborating that the DRIS makes Alice's transmission more distinguishable at Willie.
Moreover, the absolute DEP values under the DRIS are not identical, because the two schemes adopt different test statistics. 
Specifically, the proposed detector uses the accumulated energy statistic $Y_{\rm w}=\sum_{n=1}^{N}|y_{\rm w}(n)|^2$ and performs detection based on the plug-in LLR, whereas the scheme in~\cite{huang2025simultaneously} is based on per-sample decisions.
Hence, numerical differences in DEP are expected even under the same DRIS setting.

Second, the achievable rate without a DRIS increases rapidly with transmit power, whereas the achievable rate with a DRIS grows much more slowly. 
In addition, the simulated rate curve with a DRIS agrees well with the theoretical curve obtained from~\eqref{SNRBobReduceAr}, which validates our asymptotic analysis. 
As the transmit power increases, the widening gap between the two rate curves shows that a larger $P_0$ does not result in a larger communication gain in the presence of the DRIS.
This behavior is consistent with the analytical analysis: increasing $P_0$ strengthens not only the desired signal but also the DRIS-induced ACA term, thereby limiting the net SJNR improvement at Bob. 

For further comparison, the achievable rate with a DRIS remains between the two active jamming benchmarks over most of the considered power range and becomes close to the AJ benchmark with $+3$ dBm jamming power in the high transmit power regime. This is particularly noteworthy because the DRIS achieves such a jamming impact without requiring dedicated jamming power. Hence, Fig.~\ref{fig3} further demonstrates the strategic advantage of the DRIS: it simultaneously reduces Willie's DEP and suppresses Bob's achievable rate, without requiring either CSI or additional jamming power. By contrast, although the impact of an AJ increases with its jamming power, this benefit comes at the cost of persistent energy consumption. This is because the AJ must remain active since Willie does not know a priori whether Alice is transmitting. It is worth noting that the DRIS affects covert communications differently under low- and high-transmit-power regimes at Alice. Therefore, the following discussion focuses on the impact of different factors in these two regimes.
 
\begin{figure*}[!t]
    \centering
    \subfloat{
            \includegraphics[scale=0.596]{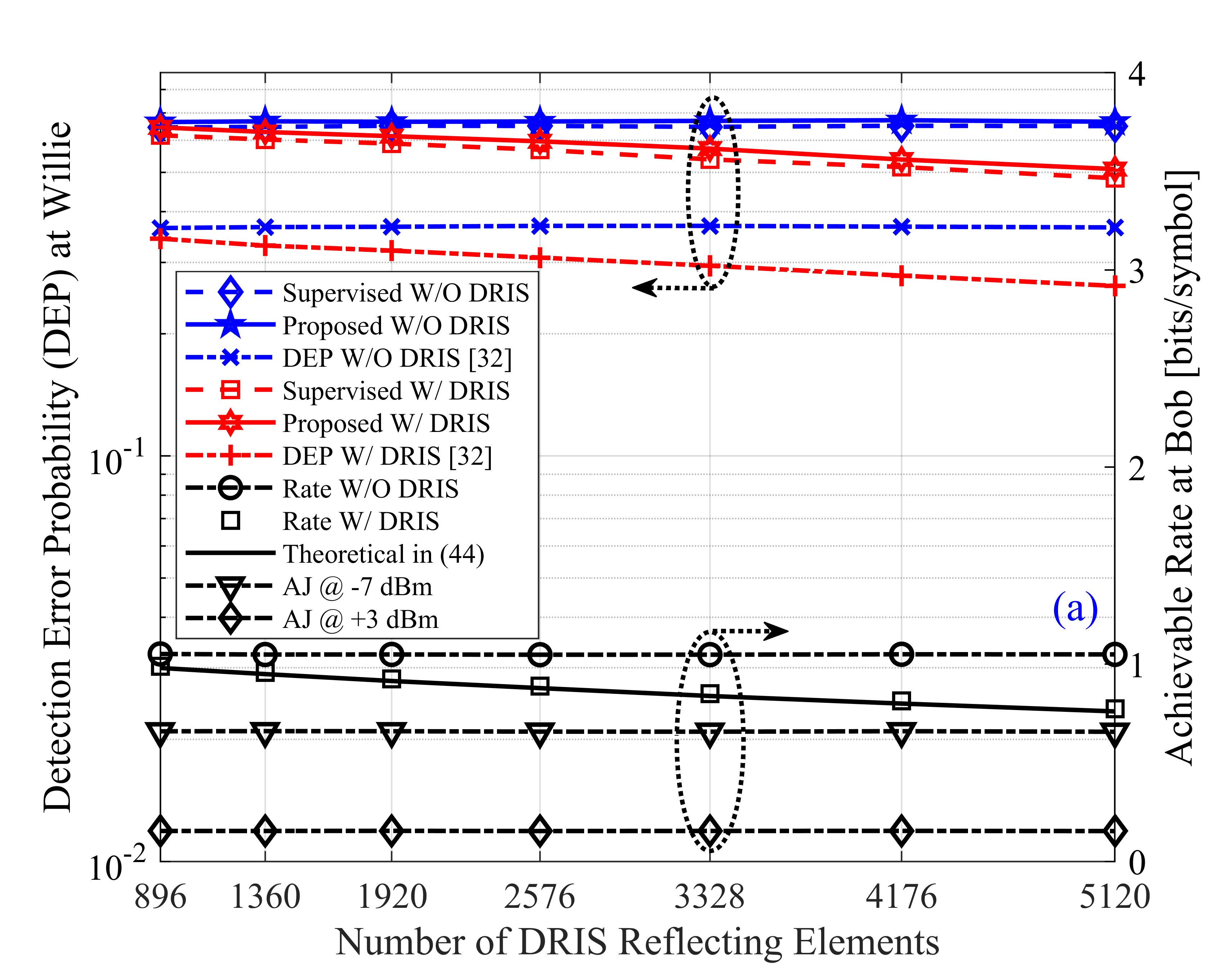}}\hspace{12pt}
    \subfloat{
            \includegraphics[scale=0.596]{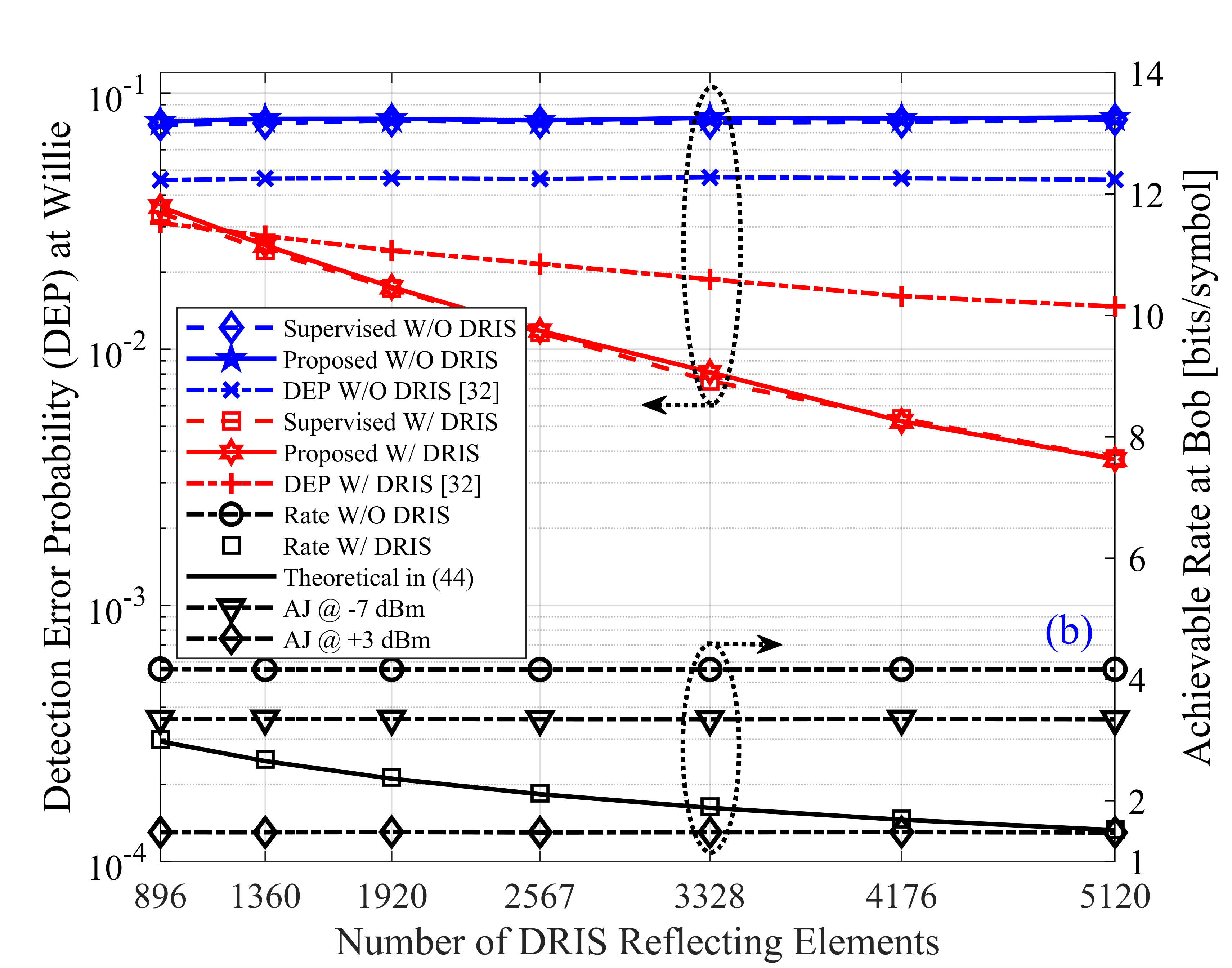}}
   \caption{Detection error probability (DEP) at Willie (left y-axis) and achievable rate at Bob (right y-axis) vs. number of DRIS reflecting elements at (a) low transmit power (-6 dBm) and (b) high transmit power (6 dBm).}
    \label{ResFigND}
\end{figure*}

\begin{figure*}[!t]
    \centering
    \subfloat{
            \includegraphics[scale=0.596]{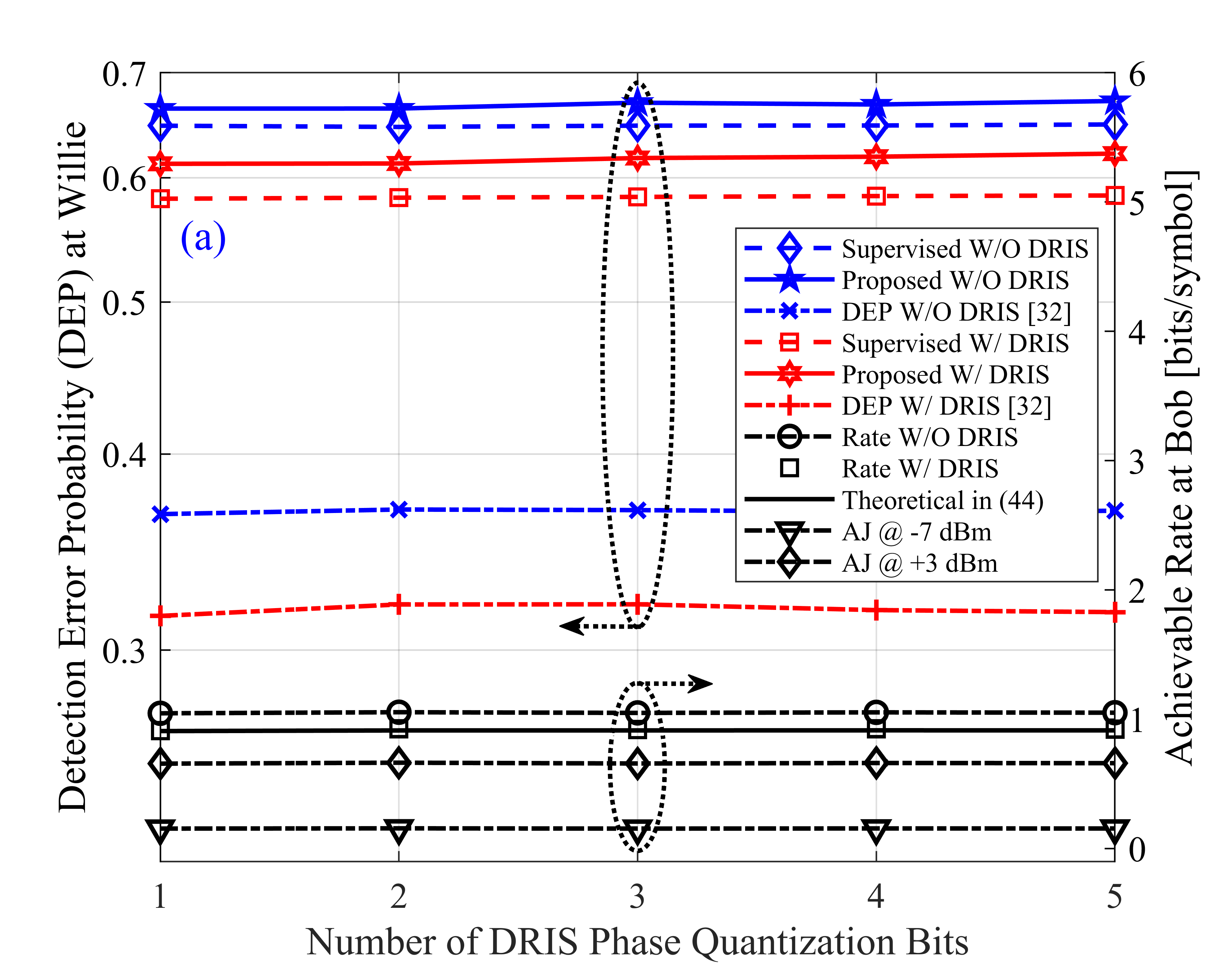}}\hspace{12pt}
    \subfloat{
            \includegraphics[scale=0.596]{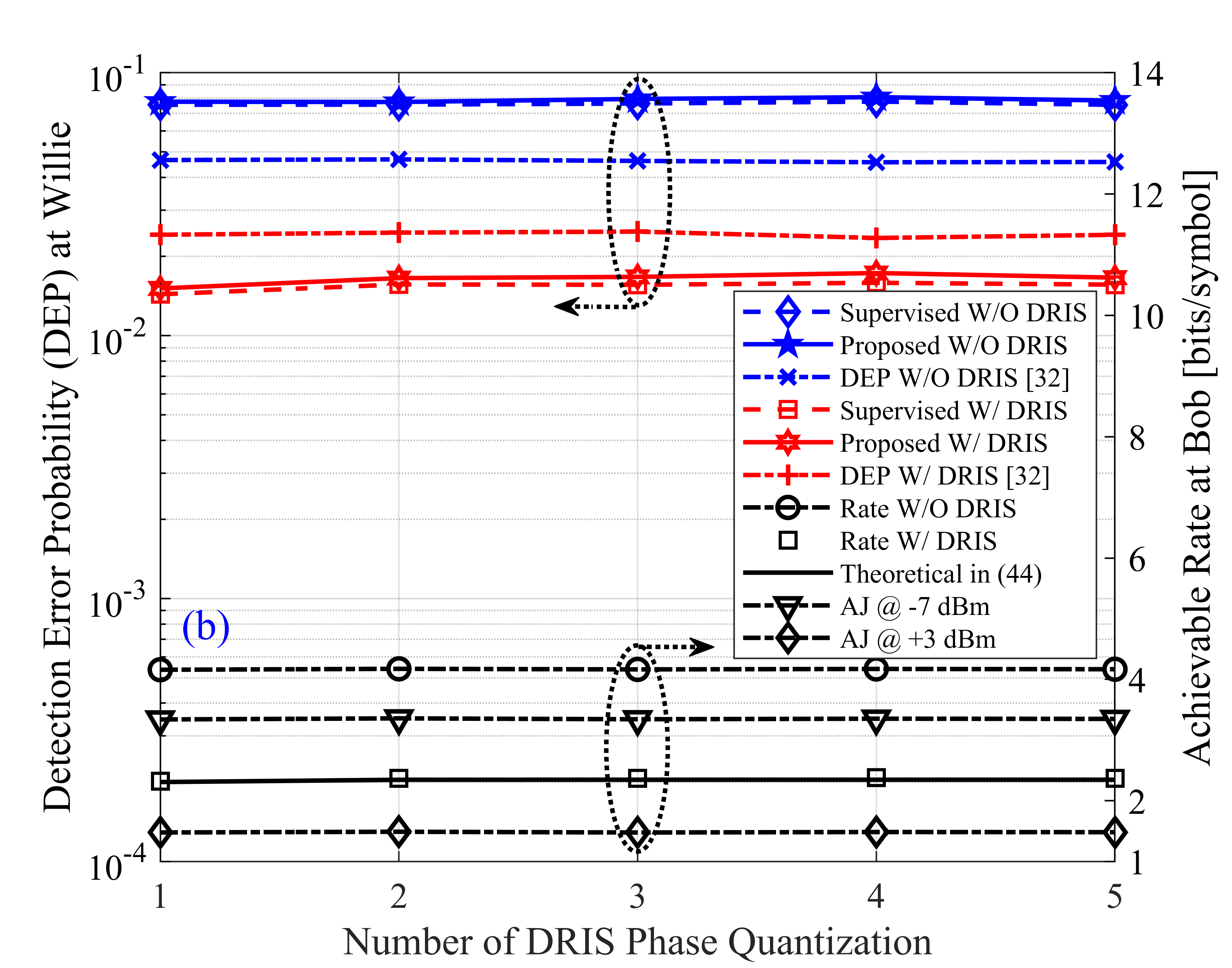}}
   \caption{Detection error probability (DEP) at Willie (left y-axis) and achievable rate at Bob (right y-axis) vs. number of DRIS phase quantization bits at (a) low transmit power (-6 dBm) and (b) high transmit power (6 dBm).}
    \label{ResFigBit}
\end{figure*}

Fig.~\ref{ResFigND} shows the DEP at Willie and the achievable rate at Bob versus the number of DRIS reflecting elements under the low-transmit-power ($-6$ dBm) and high-transmit-power ($6$ dBm) regimes at Alice. For both regimes, the DEP achieved by the proposed unsupervised detector remains close to that of its supervised counterpart, with and without a DRIS, over the entire range of reflecting elements, which confirms the effectiveness of the proposed detector under different DRIS sizes. This suggests that the proposed detector can reliably learn the analytically intractable density under ${\mathcal H}_1$ while simultaneously estimating the Gamma scale $\widehat \delta _{\rm{w}}^2$ under ${\mathcal H}_0$.

In the low-transmit-power regime shown in Fig.~\ref{ResFigND}~(a), increasing the number of DRIS reflecting elements only moderately reduces Willie's DEP and causes a relatively mild loss in Bob's achievable rate. In the high-transmit-power regime shown in Fig.~\ref{ResFigND}~(b), however, the DEP with a DRIS decreases much more rapidly and Bob's achievable rate drops significantly as the number of reflecting elements increases. This behavior is consistent with Proposition~\ref{Proposition1} and Proposition~\ref{Proposition2}, where the DRIS-induced ACA scales with $N_{\rm D}\overline{\alpha}$. Enlarging the DRIS has only a limited impact when $P_0$ is small, but becomes much more effective when $P_0$ is large.
Moreover, the simulated achievable-rate curves with a DRIS also agree well with the theoretical curves in both regimes, which validates the asymptotic analysis in Section~\ref{Bobcommun}. 

To investigate the impact of the DRIS phase-quantization bits, we consider a $b$-bit quantized DRIS with $2^b$ discrete phase-shift values. Following~\cite{PGFun1}, the phase-shift set is modeled as $\Psi_b=
\left\{-\frac{\pi}{2}, -\frac{\pi}{2}+\frac{\pi}{2^{b-1}}, \ldots, \frac{3\pi}{2}-\frac{\pi}{2^{b-1}}\right\}$.
For each DRIS reflecting element, the time-varying phase shift is uniformly drawn from $\Psi_b$, and the corresponding amplitude is determined by the practical phase-dependent hardware response
\begin{align}
\beta_n(t)
&=
{\mathcal F}\!\left(\varphi_n(t)\right) \nonumber\\
&=
(1-\alpha_{\min})
\left(
\frac{\sin\!\left(\varphi_n(t)-\phi\right)+1}{2}
\right)^{\mu}
+\alpha_{\min},
\label{PAFun_bit}
\end{align}
where $\varphi_n(t)\in\Psi_b$, $\alpha_{\min}=0.8$, $\mu=1.6$, and $\phi=0$~\cite{PGFun1}. Under equiprobable phase selection, the average squared  amplitude entering the asymptotic analysis is
$\overline{\alpha} = \frac{1}{2^b}\sum_{\varphi\in\Psi_b} {\mathcal F}^2(\varphi)$.

\begin{figure*}[!t]
    \centering
    \subfloat{
            \includegraphics[scale=0.596]{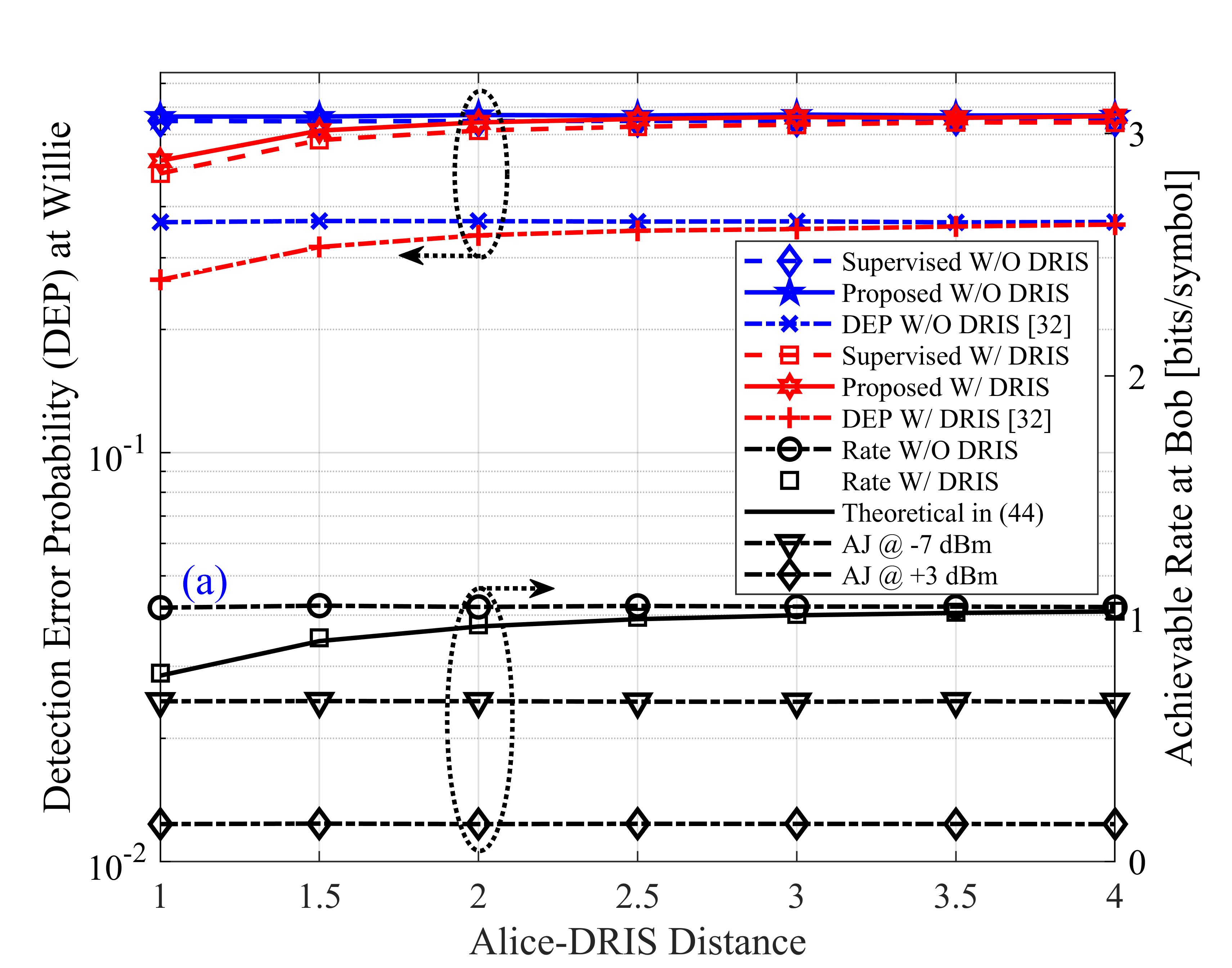}}\hspace{12pt}
    \subfloat{
            \includegraphics[scale=0.596]{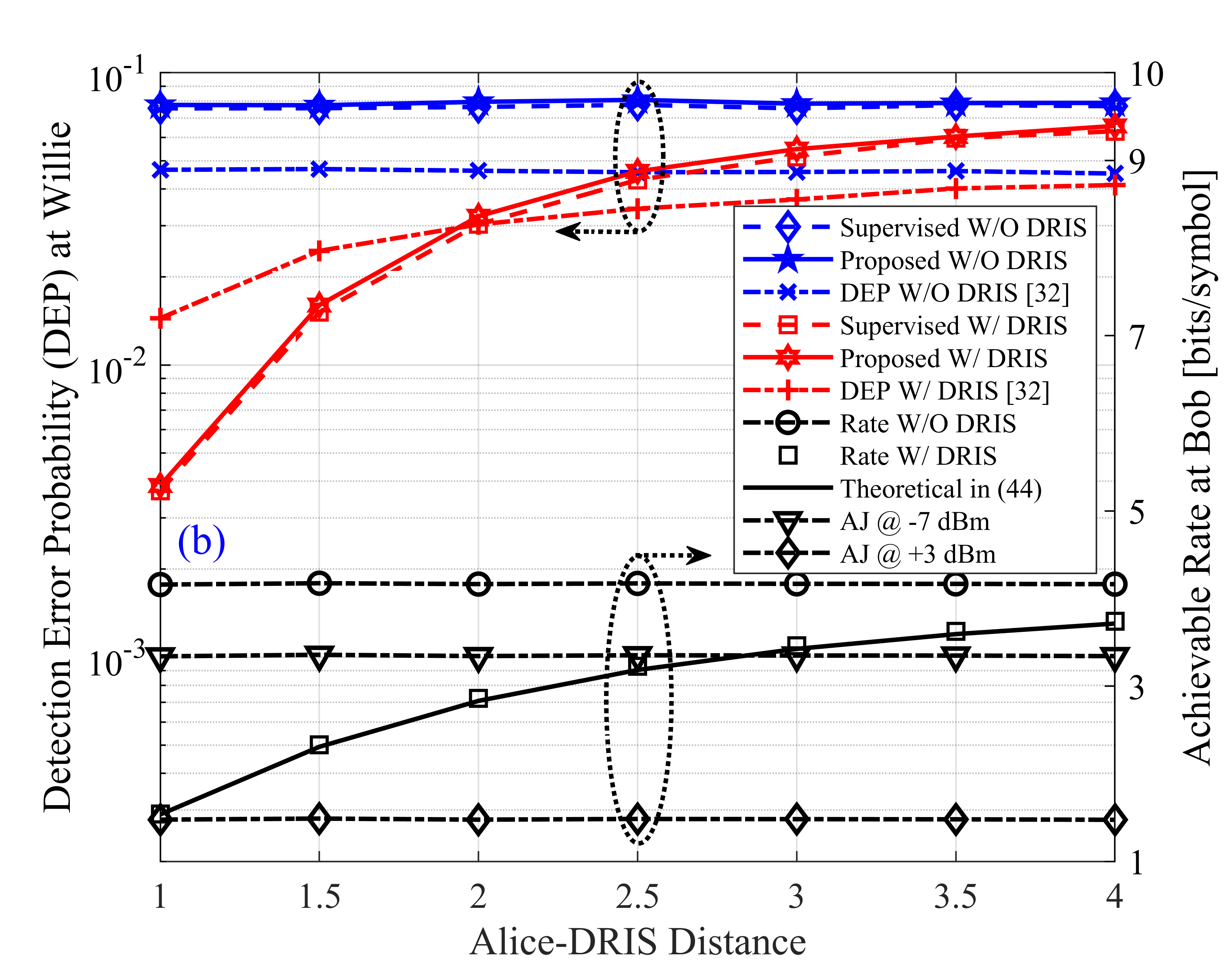}}
   \caption{Detection error probability (DEP) at Willie (left y-axis) and achievable rate at Bob (right y-axis) vs. distance between Alice and DRIS at (a) low transmit power (-6 dBm) and (b) high transmit power (6 dBm).}
    \label{ResFigDisAD}
\end{figure*}

Fig.~\ref{ResFigBit} shows the DEP at Willie and the achievable rate at Bob versus the DRIS phase-quantization resolution. For all considered quantization bits, the DEP achieved by the proposed unsupervised detector remains close to that of its supervised counterpart, with and without a DRIS, indicating stable detection performance across different DRIS phase quantization bits.
Increasing the DRIS phase-quantization resolution only slightly changes both Willie's DEP and Bob's achievable rate. In particular, the DEP with a DRIS improves only marginally as the number of quantization bits increases, while Bob's achievable rate under the DRIS decreases only slightly. 
The impact of the DRIS quantization resolution enters mainly through $\overline{\alpha}$. According to the practical phase-dependent amplitude response in~\eqref{PAFun_bit}, increasing the number of quantization bits does not significantly alter the average squared reflection amplitude. Therefore, it does not substantially strengthen the DRIS-induced ACA.  As a result, both the additional monitoring gain at Willie and the additional communication degradation at Bob caused by finer phase quantization are limited.

From Fig.~\ref{ResFigBit}, the Rate W/ DRIS curve is close to its theoretical curve, which further validates the asymptotic analysis. More importantly, the above results show that even a 1-bit phase-quantized DRIS is already sufficient to reduce Willie's DEP and suppress Bob's achievable rate. Therefore, increasing the phase-quantization resolution beyond 1 bit provides only marginal additional benefit for Willie. This result is practically significant since it shows that a low-complexity DRIS with coarse phase quantization can achieve most monitoring and jamming effects without fine-grained phase control.

Fig.~\ref{ResFigDisAD} shows the DEP at Willie and the achievable rate at Bob versus the distance between Alice and the DRIS. For both low- and high-transmit-power regimes, the DEP achieved by the proposed unsupervised detector remains close to that of its supervised counterpart over the entire distance range, indicating that the proposed detector remains effective as the Alice--DRIS distance varies.
When the DRIS is present, increasing the Alice--DRIS distance $d _ { A D }$ weakens the impact of the DRIS on both Willie and Bob. This is because a larger $d _ { A D }$ results in stronger large-scale fading and thus reduces the strength of the cascaded DRIS channels, which weakens the DRIS-induced ACA. Consequently, Willie's monitoring performance gain is reduced, as reflected by the increasing DEP with a DRIS, while Bob's achievable rate improves as the DRIS-induced ACA becomes weaker.  

\section{Conclusions}\label{conclusions}
In this paper, we investigated covert communications in the presence of a DRIS and developed a novel unsupervised semi-parametric plug-in likelihood-ratio detector for Willie. The proposed detector retains a parametric Gamma reference model under ${\mathcal H}_0$ with an unknown scale parameter, while learning the analytically intractable density under ${\mathcal H}_1$ from unlabeled data through a one-dimensional monotone normalizing flow model. The results show that the DRIS can simultaneously improve Willie's monitoring performance and degrade the communication performance between Alice and Bob, without requiring either CSI or additional jamming power. Meanwhile, the proposed unsupervised detector does not rely on labeled data, which makes it well suited to adversarial covert-communication scenarios. The main conclusions are summarized as follows.

1) The proposed unsupervised semi-parametric plug-in likelihood-ratio detector enables Willie to perform effective monitoring even when neither the density under ${\mathcal H}_1$ nor the noise variance $\delta^2_{\rm w}$ under ${\mathcal H}_0$ is known a priori. By exploiting the structural prior that the test statistic under ${\mathcal H}_0$ (i.e., pure noise) follows a Gamma distribution, the prposed unsupervised detector jointly estimates the Gamma scale and learns the intractable density under ${\mathcal H}_1$ directly from unlabeled observations. Simulation results show that the proposed unsupervised detector achieves monitoring performance close to that of its supervised counterpart.

2) The DRIS, which has random and time-varying reflection configurations, introduces ACA, thereby creating a strategic advantage for Willie. Specifically, it lowers Willie's DEP and suppresses Bob's achievable rate without requiring CSI or additional jamming power. Moreover, to maximize the DRIS-induced impact, it is advantageous for Willie to deploy the DRIS close to Alice, since a shorter Alice--DRIS distance strengthens the cascaded DRIS channels and intensifies the resulting ACA.

3) Due to the DRIS-induced ACA, increasing Alice's transmit power does not significantly improve Bob's communication performance. Instead, it strengthens the DRIS-induced impairment and increases Alice's risk of detection by Willie. In addition, the results show that expanding the DRIS is effective. However, increasing the DRIS phase quantization bits beyond one bit only provides a marginal benefit. Hence, even a low-complexity DRIS with 1-bit phase quantization is sufficient to substantially improve the Willie's monitoring performance and degrade the communication performance.

Overall, this work demonstrates that a DRIS, together with the proposed unsupervised semi-parametric plug-in likelihood-ratio detector, can effectively monitor and jam covert communications without requiring labeled data or additional jamming power. One interesting direction for future work is to investigate practical countermeasures at Alice and Bob, such as anomaly-aware transmission adaptation, environment sensing, and DRIS-aware robust communication strategies.

\begin{appendices}
\section{Proof of Proposition~\ref{Proposition1}}\label{AppendixA}
Based on the cascaded DRIS channel between Alice and Willie expressed in~\eqref{hDEx} and 
the definition that 
$\varphi(t) = \big[\beta_{1}(t)e^{j\varphi_{1}(t)}, \cdots, \beta_{N_{D}}(t)e^{j\varphi_{N_{D}}(t)}\big]$, 
we can rewrite ${ h}^{\rm w}_{\rm D}(t) $ as

\begin{equation}
\begin{split}  
{ h}^{\rm w}_{\rm D}(t)
= \!\sqrt{\!\!\frac{\kappa _{g}}{(\kappa _{g}+1){\mathcal{L}^{{{{\nu _{{\rm{g}}}}}}}}{\mathcal{L}^{{{{\nu^{\rm w}_{{\rm{I}} }}}}}} }}
\!\sum_{r=1}^{N_{D}}\!\big[{{\widehat{{\boldsymbol g }}}}^{\mathrm{LOS}}\big]_{r}\!\big[{\widehat{\boldsymbol{h}}^{{\rm w}}_{{\rm I}} }\big]_{r}
\beta_{r}(t)e^{j\varphi_{r}(t)}\\
+ \!\sqrt{\!\!\frac{1}{(\kappa _{g}+1){\mathcal{L}^{{{{\nu _{{\rm{g}}}}}}}}{\mathcal{L}^{{{{\nu^{\rm w}_{{\rm{I}} }}}}}}}}
\!\sum_{r=1}^{N_{D}}\!\big[{{\widehat{{\boldsymbol g }}}}^{\mathrm{NLOS}}\big]_{r}\!\big[{\widehat{\boldsymbol{h}}^{{\rm w}}_{{\rm I}} }\big]_{r}
\beta_{r}(t)e^{j\varphi_{r}(t)}.
\label{eq:A2}
\end{split}
\end{equation}

Assuming that ${\widehat{\boldsymbol{h}}^{{\rm w}}_{{\rm I}} }$, $\varphi(t)$, and ${{\widehat{{\boldsymbol g }}}}$ consist of independent i.i.d. elements
and $\mathbb{E}[{\widehat{\boldsymbol{h}}^{{\rm w}}_{{\rm I}} }]=0$,
the expectations of the terms in~\eqref{eq:A2} are obtained as follows:
\begin{equation}
\begin{split} 
&\mathbb{E}\!\left[\big[{{\widehat{{\boldsymbol g }}}}^{\mathrm{LOS}}\big]_{r}\big[{\widehat{\boldsymbol{h}}^{{\rm w}}_{{\rm I}} }\big]_{r}
\beta_{r}(t)e^{j\varphi_{r}(t)}\right] \\ & =\mathbb{E}\!\big[{{\widehat{{\boldsymbol g }}}}^{\mathrm{LOS}}\beta_{r}(t)e^{j\varphi_{r}(t)}]_{r}
\mathbb{E}\!\big[{\widehat{\boldsymbol{h}}^{{\rm w}}_{{\rm I}} }\big]_{r}
= 0, \label{eq:A3}\\
\end{split}
\end{equation}
and
\begin{equation}
\begin{split} 
&\mathbb{E}\!\left[\big[{{\widehat{{\boldsymbol g }}}}^{\mathrm{NLOS}}\big]_{r}\big[{\widehat{\boldsymbol{h}}^{{\rm w}}_{{\rm I}} }\big]_{r}
\beta_{r}(t)e^{j\varphi_{r}(t)}\right] \\ &=\mathbb{E}\!\big[{{\widehat{{\boldsymbol g }}}}^{\mathrm{NLOS}}\beta_{r}(t)e^{j\varphi_{r}(t)}]_{r}
\mathbb{E}\!\big[{\widehat{\boldsymbol{h}}^{{\rm w}}_{{\rm I}} }\big]_{r}
= 0. \label{eq:A4}
\end{split}
\end{equation}
Furthermore, the corresponding variances are expressed as
\begin{equation}
\mathrm{Var}\!\left[\!\big[{{\widehat{{\boldsymbol g }}}}^{\mathrm{LOS}}\!\big]_{r}\big[{\widehat{\boldsymbol{h}}^{{\rm w}}_{{\rm I}} }\big]_{r}
\beta_{r}(t)e^{j\varphi_{r}(t)}\right]
\!\!=\!\! \mathbb{E}\!\left[\big|\big[{\widehat{\boldsymbol{h}}^{{\rm w}}_{{\rm I}} }\big]_{r}\big|^{2}\!\right]
\mathbb{E}\!\left[\big|\beta_{r}(t)\big|^{2}\right],
\label{eq:A5}
\end{equation}
where 
\begin{equation}
\mathbb{E}\!\left[\big|\beta_{r}(t)\big|^{2}\right] = \sum_{i=1}^{2^{b}}P_{i}\alpha_{i}^{2} = \overline{\alpha},
\label{eq:A51}
\end{equation}
 $P_{i}$ denotes the probability that the DRIS phase shift $\varphi_{r}(t)$ 
takes the $i$-th value $\phi_{i}\in\Omega$, i.e.,
$P_{i}=\mathbb{P}(\varphi_{r}(t)=\phi_{i}),\ \forall r$.

Consequently, the variances in~\eqref{eq:A5} can be obtained as
\begin{equation}
\mathrm{Var}\!\left[\big[{{\widehat{{\boldsymbol g }}}}^{\mathrm{LOS}}\big]_{r}\big[{\widehat{\boldsymbol{h}}^{{\rm w}}_{{\rm I}} }\big]_{r}
\beta_{r}(t)e^{j\varphi_{r}(t)}\right]
= \overline{\alpha},
\label{eq:A52}
\end{equation}

Take the same steps as above, we can obtain that
\begin{equation}
\begin{split}
&\mathrm{Var}\!\left[\big[{{\widehat{{\boldsymbol g }}}}^{\mathrm{NLOS}}\big]_{r}\big[{\widehat{\boldsymbol{h}}^{{\rm w}}_{{\rm I}} }\big]_{r}
\beta_{r}(t)e^{j\varphi_{r}(t)}\right]\\
&= \mathbb{E}\!\left[\big|\big[{{\widehat{{\boldsymbol g }}}}^{\mathrm{NLOS}}\big]_{r}\big|^{2}\right]
\mathbb{E}\!\left[\big|\big[{\widehat{\boldsymbol{h}}^{{\rm w}}_{{\rm I}} }\big]_{r}\big|^{2}\right]
\mathbb{E}\!\left[\big|\beta_{r}(t)\big|^{2}\right] 
=\overline{\alpha}.
\label{eq:A6}
\end{split}
\end{equation}

When $N_{D}$ is sufficiently large, according to the Lindeberg--Lévy central limit theorem, 
the normalized sum converges in distribution to a circularly-symmetric 
complex Gaussian random variable, i.e., 
\begin{equation}
       h_{\rm D}^{\rm w}(t)=  { {{\frac{{{\widehat{{\boldsymbol g }}}}
     {{\rm{diag}}\!\left({\boldsymbol \varphi }(t) \right)} {\widehat{\boldsymbol{h}}^{{\rm w}}_{{\rm I}} } }{{  {\mathcal{L}^{\frac{{{\nu_{{\rm{g}}}}}}{2}}} {\mathcal{L}^{\frac{{{\nu^{\rm{w}}_{{\rm{I}}}}}}{2}}} } }}} }
         \mathop  \to \limits^{\rm{d}} \mathcal{CN}\!\left( {0,  {\frac{ N_{\rm D}\overline \alpha}{{ {\mathcal{L}^{{{{\nu _{{\rm{g}}}}}}}} {\mathcal{L}^{{{{\nu^{\rm w}_{{\rm{I}} }}}}}} } }} } \right).
        \label{eq:A7}
    \end{equation}

\section{Proof of Proposition~\ref{Proposition2}}\label{AppendixB}
The proof is analogous to that in Appendix~\ref{AppendixA}.
\end{appendices}

\end{document}